\documentclass[12pt]{article}

\usepackage[margin=1in]{geometry}
\usepackage{marginnote}
\usepackage[toc,page,title]{appendix}									
\usepackage{amsmath}													
\usepackage{amsfonts}													
\usepackage{amssymb}													
\usepackage{amsthm}														
\usepackage{amstext}
\usepackage{mathtools}
\usepackage{bm}															
\usepackage{booktabs}
\usepackage{adjustbox}
\usepackage[center,sc,hang]{caption}											
\usepackage{fancybox}													
\usepackage[OT4,OT1,T2A,T5,T1]{fontenc}									
\usepackage{tikz,pgf,graphicx}
\usepackage{color}
\usepackage{lscape}
\usepackage{makeidx}
\usepackage[sc]{mathpazo}
\usepackage{multirow,subfig}
\usepackage{setspace}
\usepackage[multiple]{footmisc}
\usepackage{microtype}
\usepackage{enumitem}
\usepackage{pstricks}
\usepackage{pstricks-add}
\usepackage{pst-bar}
\usepackage{pst-node}
\usepackage{pst-text}
\usepackage{pst-3d}
\usepackage{pst-coil}
\usepackage{pst-plot}
\usepackage{pst-lens}
\usepackage{rotating}
\usepackage{sectsty}
\usepackage{sgame}
\usepackage{soul}
\usepackage{textcomp}
\usepackage[sort]{natbib}
\usepackage{authblk}
\usepackage[hyperfootnotes=false,citecolor=Blue,colorlinks=true]{hyperref}

\psset{xunit=1in,yunit=1in,runit=1in}

\allsectionsfont{\centering\mdseries\scshape}
\captionlabelfont{\mdseries\scshape}
\renewcommand{\abstractname}{\large\mdseries\scshape Abstract}

\newtheorem*{thm*}{Theorem}
\newtheorem{lemma}{Lemma}
\newtheorem{proposition}{Proposition}

\newtheorem{ass}{Assumption}

\renewcommand{\geq}{\geqslant}
\renewcommand{\leq}{\leqslant}
\newcommand{\wc}{\mkern 2mu\cdot\mkern 2mu}
\newcommand{\od}{\mathop{}\!\mathrm{d}}
\newcommand{\BB}{\mathit{BB}}
\newcommand{\BW}{\mathit{BW}}

\newcommand{\WW}{\mathit{WW}}
\newcommand{\one}{\mathit{C}_1}
\newcommand{\two}{\mathit{C}_2}
\newcommand{\urnone}{\text{urn 1}}
\newcommand{\urntwo}{\text{urn 2}}

\usepackage{accents}
\newcommand{\ubar}[1]{\underaccent{\bar}{#1}}
\DeclareMathOperator{\ibeta}{I}
\DeclareMathOperator{\argmax}{argmax}
\DeclareMathOperator{\Beta}{B}
\DeclareMathOperator{\var}{Var}
\allowdisplaybreaks

\makeatletter

\newcommand\listoftodos{\clearpage\section*{Todo list} \@starttoc{tdo}}
\newcommand\l@todo[2] {\par\noindent {$\Box$\;#1\dotfill\makebox[1.5em][r]{#2}\;}\par}
\makeatother

\definecolor{PastelYellow}{rgb}{0.99,.96,0.85}
\definecolor{PastelBlue}  {rgb}{0.88,.89,0.98}
\definecolor{LightYellow} {rgb}{1.,1.,0.88}
\definecolor{PaleGreen}   {rgb}{0.6,0.98,0.6}
\definecolor{PaleGreenB}  {rgb}{0.9,1,0.9}
\definecolor{green1}      {rgb}{0,0.2,0}
\definecolor{green2}      {rgb}{0,0.45,0}
\definecolor{green3}      {rgb}{0,0.7,0}
\definecolor{green4}      {rgb}{0,0.95,0}
\definecolor{LightGray}   {gray}{0.92}
\definecolor{YellowOrange}{cmyk}{0,0.42,1,0}
\definecolor{OliveGreen}  {cmyk}{0.64,0,0.95,0.40}
\definecolor{GreenYellow}   {cmyk}{0.15,0,0.69,0}
\definecolor{Yellow}        {cmyk}{0,0,1,0}
\definecolor{Goldenrod}     {cmyk}{0,0.10,0.84,0}
\definecolor{Dandelion}     {cmyk}{0,0.29,0.84,0}
\definecolor{Apricot}       {cmyk}{0,0.32,0.52,0}
\definecolor{Peach}         {cmyk}{0,0.50,0.70,0}
\definecolor{Melon}         {cmyk}{0,0.46,0.50,0}
\definecolor{YellowOrange}  {cmyk}{0,0.42,1,0}
\definecolor{Orange}        {cmyk}{0,0.61,0.87,0}
\definecolor{BurntOrange}   {cmyk}{0,0.51,1,0}
\definecolor{Bittersweet}   {cmyk}{0,0.75,1,0.24}
\definecolor{RedOrange}     {cmyk}{0,0.77,0.87,0}
\definecolor{Mahogany}      {cmyk}{0,0.85,0.87,0.35}
\definecolor{Maroon}        {cmyk}{0,0.87,0.68,0.32}
\definecolor{BrickRed}      {cmyk}{0,0.89,0.94,0.28}
\definecolor{Red}           {cmyk}{0,1,1,0}
\definecolor{OrangeRed}     {cmyk}{0,1,0.50,0}
\definecolor{RubineRed}     {cmyk}{0,1,0.13,0}
\definecolor{WildStrawberry}{cmyk}{0,0.96,0.39,0}
\definecolor{Salmon}        {cmyk}{0,0.53,0.38,0}
\definecolor{CarnationPink} {cmyk}{0,0.63,0,0}
\definecolor{Magenta}       {cmyk}{0,1,0,0}
\definecolor{VioletRed}     {cmyk}{0,0.81,0,0}
\definecolor{Rhodamine}     {cmyk}{0,0.82,0,0}
\definecolor{Mulberry}      {cmyk}{0.34,0.90,0,0.02}
\definecolor{RedViolet}     {cmyk}{0.07,0.90,0,0.34}
\definecolor{Fuchsia}       {cmyk}{0.47,0.91,0,0.08}
\definecolor{Lavender}      {cmyk}{0,0.48,0,0}
\definecolor{Thistle}       {cmyk}{0.12,0.59,0,0}
\definecolor{Orchid}        {cmyk}{0.32,0.64,0,0}
\definecolor{DarkOrchid}    {cmyk}{0.40,0.80,0.20,0}
\definecolor{Purple}        {cmyk}{0.45,0.86,0,0}
\definecolor{Plum}          {cmyk}{0.50,1,0,0}
\definecolor{Violet}        {cmyk}{0.79,0.88,0,0}
\definecolor{RoyalPurple}   {cmyk}{0.75,0.90,0,0}
\definecolor{BlueViolet}    {cmyk}{0.86,0.91,0,0.04}
\definecolor{Periwinkle}    {cmyk}{0.57,0.55,0,0}
\definecolor{CadetBlue}     {cmyk}{0.62,0.57,0.23,0}
\definecolor{CornflowerBlue}{cmyk}{0.65,0.13,0,0}
\definecolor{MidnightBlue}  {cmyk}{0.98,0.13,0,0.43}
\definecolor{NavyBlue}      {cmyk}{0.94,0.54,0,0}
\definecolor{NB1}           {cmyk}{.85,1,0,0}
\definecolor{NB2}           {cmyk}{.70,.9,0,0}
\definecolor{NB3}           {cmyk}{.55,.8,0,0}
\definecolor{RoyalBlue}     {cmyk}{1,0.50,0,0}
\definecolor{Blue}          {cmyk}{1,1,0,0}
\definecolor{Cerulean}      {cmyk}{0.94,0.11,0,0}
\definecolor{Cyan}          {cmyk}{1,0,0,0}
\definecolor{ProcessBlue}   {cmyk}{0.96,0,0,0}
\definecolor{SkyBlue}       {cmyk}{0.62,0,0.12,0}
\definecolor{Turquoise}     {cmyk}{0.85,0,0.20,0}
\definecolor{TealBlue}      {cmyk}{0.86,0,0.34,0.02}
\definecolor{Aquamarine}    {cmyk}{0.82,0,0.30,0}
\definecolor{BlueGreen}     {cmyk}{0.85,0,0.33,0}
\definecolor{Emerald}       {cmyk}{1,0,0.50,0}
\definecolor{JungleGreen}   {cmyk}{0.99,0,0.52,0}
\definecolor{SeaGreen}      {cmyk}{0.69,0,0.50,0}
\definecolor{Green}         {cmyk}{1,0,1,0}
\definecolor{ForestGreen}   {cmyk}{0.91,0,0.88,0.12}
\definecolor{PineGreen}     {cmyk}{0.92,0,0.59,0.25}
\definecolor{LimeGreen}     {cmyk}{0.50,0,1,0}
\definecolor{YellowGreen}   {cmyk}{0.44,0,0.74,0}
\definecolor{SpringGreen}   {cmyk}{0.26,0,0.76,0}
\definecolor{OliveGreen}    {cmyk}{0.64,0,0.95,0.40}
\definecolor{RawSienna}     {cmyk}{0,0.72,1,0.45}
\definecolor{Sepia}         {cmyk}{0,0.83,1,0.70}
\definecolor{Brown}         {cmyk}{0,0.81,1,0.60}
\definecolor{Tan}           {cmyk}{0.14,0.42,0.56,0}
\definecolor{Gray}          {cmyk}{0,0,0,0.50}
\definecolor{Black}         {cmyk}{0,0,0,1}
\definecolor{White}         {cmyk}{0,0,0,0}

\begin{document}
\pagenumbering{gobble}
\title{Belief Error and Non-Bayesian Social Learning: Experimental Evidence  \thanks{Bo\u{g}a\c{c}han \c{C}elen: bc319@nyu.edu, Sen Geng: gengsen2020@gmail.com, Huihui Li: huihui.pku@gmail.com. We have benefited from the helpful comments of seminar participants at Huazhong University of Science and Technology, Monash University, and Xiamen University, and participants at the conferences of  Xiamen International Workshop on Experimental Economics (2014, 2016), China Meeting of Econometric Society (2017), NYU CESS Meeting (2017), China Greater Bay Area Experimental Economics Workshop (2018), ESA World Meeting (2018). Geng acknowledges financial support from the Fundamental Research Funds for the Central Universities (No. 20720151323) and China's NSF Grant (No. 71703134). }}

\author[1]{Bo\u{g}a\c{c}han \c{C}elen }
\affil[1]{University of Melbourne}
\author[2]{Sen Geng}
\author[2]{Huihui Li}
\affil[2]{Xiamen University}
\maketitle

\begin{abstract}
This paper experimentally studies whether individuals hold a first-order belief that others apply Bayes' rule to incorporate private information into their beliefs, which is a fundamental assumption in many Bayesian and non-Bayesian social learning models. We design a novel experimental setting in which the first-order belief assumption implies that social information is equivalent to private information. Our main  finding is that participants' reported reservation prices of social information are significantly lower than those of private information, which provides evidence that casts doubt on the first-order belief assumption. We also build a novel belief error model in which participants form a random posterior belief with a Bayesian posterior belief kernel to explain the experimental findings. A structural estimation of the model suggests that participants' sophisticated consideration of others' belief error and their exaggeration of the error both contribute to the difference in reservation prices.\\


\end{abstract}

\newpage
\pagenumbering{arabic}

\section{Introduction \label{sec:intro}}
In many settings where individuals have to make a choice without knowing the underlying state that is payoff relevant, they can learn from either observing signals indicative of the underlying state (labeled as private information) or observing predecessors' choices in similar settings (labeled as social information).\footnote{This article confines its discussion of social information to observing others' actions. More generally, social information also includes observing payoffs of others' actions, communication and observation of others' beliefs and opinions.}   A typical social learning environment often involves both private information and social information. How individuals aggregate the two types of information is a continuing focus in the literature on social learning.

As a useful benchmark, Bayesian social learning models assume that it is common knowledge that individuals apply Bayes' rule to incorporate private information and social information into their beliefs.\footnote{See \citet{bikhchandani:92}, \citet{banerjee:92}, \citet{smith:00} for an exploration of information aggregation in settings where each individual observes private information and all predecessors' choices. See \citet{acemoglu:11} for an exploration of information aggregation in settings where each individual observes private information and a stochastically generated subset of predecessors' choices.} Nevertheless, social information often has nested relationships in the sense that predecessors may base their choices on private information and social information, which itself may also originate from private information and social information, and so on. The complexity of the network structure underlying social information and the ignorance about predecessors' information structures are likely to prevent individuals from rationally interpreting social information. Motivated by this observation, non-Bayesian social learning models deviate from the Bayesian benchmark mainly in the direction of relaxing the assumption that individuals incorporate social information in a Bayesian manner. Meantime, non-Bayesian models maintain the assumption that it is common knowledge that individuals apply Bayes' law to incorporate private information, i.e., individuals incorporate private information in a Bayesian manner and hold first-order and higher-order beliefs about others' processing private information in a Bayesian way.

We design a novel experiment to test whether the first-order belief assumption that individuals think that others apply Bayes' law to incorporate private information into their beliefs holds. The multiplayer social learning game in our experiment consists of three stages. In stage 1, either urn 1 or urn 2 is selected for the game with an equal chance. Urn 1 contains $p$ fraction of black balls and $1-p$ fraction of white balls, and urn 2 contains $p$ fraction of white balls and $1-p$ fraction of black balls. The fraction $p$ is public information, but the players do not know which of the two urns is selected. Each player draws a ball from the urn with replacement and privately observes the color of the drawn ball. Then, each player makes an initial binary choice, i.e., guesses whether urn 1 or urn 2 is used. In stage 2, each player receives an endowment and states her reservation price for receiving  either additional private information (i.e., the colors of additional ball draws) or additional social information (i.e., the initial choices made by other players). In stage 3, the Becker-DeGroot-Marschak method \citep{becker:64} is employed to determine whether a player successfully purchases the additional information. Conditional on successful purchase, the player observes the additional information and then makes a second binary choice between urn 1 and urn 2. The payoff of each player is a fixed reward for her correct final choice in addition to her remaining endowment. Our experiment varies three parameters: (1) information type, i.e., additional private information or additional social information, (2) signal quantity, i.e., one or three additional observations, and (3) signal quality, with $p$ increasing from 0.6 to 0.75 to 0.9.

Given our experimental design, the first-order belief assumption implies that social information is identical to private information. Since both the network structure and information generating process are sufficiently simple and are known to individuals, whether an individual views social information as identical to private information essentially depends on whether she believes that others make a Bayesian first choice, i.e., choose urn 1 (or urn 2) after observing a black (or white) ball. In addition to the first-order belief assumption that others apply Bayes' law, two other independent assumptions work together to make her believe that others make a Bayesian first choice. The first is that she believes that no unobservables prevent her and others from agreeing on the state-contingent payoff of each action. While the possibility of such unobservables cannot be completely ruled out, it is reasonable for her to accept the absence of such a possibility given the experimental design. The second is that she believes that others choose optimally given their own beliefs. Since choosing optimally given one's own belief is treated as the basic individual rationality and is used as the implicit guiding principle underlying incentivized economic experiments, an individual in our simple decision task is expected to hold this belief. Therefore, the first-order belief assumption that others apply Bayes' law to incorporate private information is identified with a testable implication that social information has the same value as private information.

Our main finding is that participants' reservation prices for observing additional social information are significantly lower than for observing additional private information with respect to both the mean and empirical distributions of the reservation prices. This finding clearly rejects the null hypothesis about the equivalence of the two types of information and in turn casts doubt on the first-order belief assumption of social learning models. We also find that participants do not always make a Bayesian first choice and, more interestingly, that the frequency of their making Bayesian choices increases as the signal quality increases. This finding casts doubt on the assumption that individuals apply Bayes' law to incorporate private information and suggests that signal quality impacts whether participants follow Bayes' law. In addition, we find that participants generally report higher reservation prices than predicted based on the Bayesian assumption and, in particular, report a considerably positive reservation price for observing one additional signal, which has zero informational value in the Bayesian paradigm.\footnote{A few experimental studies (\citealt{bohm:97}; \citealt{plott:05}; \citealt{cason:14}) show that participants may misperceive the Becker-DeGroot-Marschak elicitation method, that the elicited value may be sensitive to the choice of bounds of the randomly generated number, and that participants may overbid. The concern about the elicitation method may dampen the last finding of the considerably higher reservation price. Nevertheless, this concern should not adversely affect our main finding of the gap in reservation prices because its effect (if any) on the treatment of private information and social information should be similar.}

A key question about the main finding is whether it can be fully explained by participants' sophisticated consideration of others' failure to apply Bayes' law to incorporate private information, i.e., whether the second finding rationalizes the first finding. To address this question and gain deeper insight into these findings, we propose a novel belief error model with two main assumptions. The first assumption is that after receiving information, participants form a random posterior belief with a Bayesian kernel  due to belief error. Specifically, we assume that the value of a participant's random posterior belief follows a beta distribution with a mean equal to the corresponding Bayesian posterior probability and with a nonnegative parameter $\gamma$ measuring the degree of belief error. The model with the assumption predicts that the chance of participants' following Bayes' law rises as the signal quality increases and that there is positive informational value of observing one additional signal. The second assumption is that participants hold the first-order belief that others also form a random posterior belief that follows a beta distribution with  a mean equal to the corresponding Bayesian posterior probability and with a nonnegative parameter $\theta$ measuring the degree of others' belief error in their opinion. From this perspective, the social information setting degenerates into the private information setting when $\theta$ shrinks to zero. Overall, the belief error model predicts that the optimal reservation price is increasing in the value of the posterior belief in the first stage on the interval $[0,\frac{1}{2}]$ but decreasing on the interval $[\frac{1}{2},1]$, indicating that the more uncertain participants feel about the underlying state, the more highly they value additional information.

We employ maximum likelihood estimation to estimate a heterogeneous belief error model. We find that the average first-order belief about others' belief error, $\bar \theta$, is considerably higher than the average belief error $\bar \gamma$. Additional generalized likelihood ratio tests confirm the difference.
Our tests reject $\bar\theta = 0$ in favor of $\bar\theta > 0$, which
suggests that a sophisticated consideration of others' belief error contributes to the gap in reservation prices of private information and social information. In addition, the hypothesis that participants' belief about others' belief error does not exceed the actual average belief error, i.e., $\bar \theta \leq \bar \gamma$, is also rejected, which suggests that an exaggeration of others' belief error also contributes to the gap in reservation prices. 

Our paper is related to the experimental studies on social learning games initiated by \citet{anderson:97}.  Subsequent experimental studies modify the baseline design and investigate systematic choice behavior off the Bayesian equilibrium path.\footnote{See \citet{hung:01}, \citet{noth:03}, \citet{kubler:04}, and \citet{goeree:07} for example.} Among these studies, \citet{noth:03}, \citet{celen:04}, and \citet{goeree:07} find that private information is overweighted relative to social information based on their model estimation. To circumvent confounding factors, such as specific modeling of the decision process, underlying the interpretation, \citet{weizsacker:10} applies a reduced-form approach to perform a meta analysis of 13 social learning experiments and finds that participants are much less likely to choose optimally in cases where the empirically optimal action contradicts their own signal. Nevertheless, the interpretation of this finding as evidence of participants' overweighting private information is undermined by the varying lucrativeness of choosing optimally.\footnote{One may define the lucrativeness of choosing optimally in any given information set as the increment in expected payoff from choosing optimally rather than choosing sub-optimally or as the expected payoff from choosing the optimal action due to the binary choice and the normalized payoff. In the meta analysis, this value is equal to the fraction of decision rounds with an underlying true state in favor of the optimal action of all decision rounds that include a specific information set. Then, a case where following one's own signal is empirically optimal corresponds to an information set in which the fraction of decision rounds with an underlying true state in favor of one's own signal is greater than one-half. A case where contradicting one's own signal is empirically optimal corresponds to an information set in which the fraction of decision rounds with an underlying true state against one's own signal is greater than one-half. Conceivably, the fraction in the former case is on average larger than in the latter case. Thus, the finding in \citet{weizsacker:10} that participants respond to incentives suggests that participants' greater reluctance to contradict their own signal may be due to the smaller incentive rather than their overweighting of private information.} Our experimental design can be viewed as a truncated version of the sequential social learning design with only the first two positions. The experiment is intentionally designed in this way to identify the effect of individuals' first-order beliefs about others' applying Bayes' law. Our paper contributes to this strand of literature by first providing unequivocal evidence that individuals value private information more than social information and identify their first-order beliefs about others' not applying Bayes' law as a reason. In addition, a recent study by \citet{defilippis:17} investigates the first two positions of a sequential social learning game. They collect  belief data and infer from the model estimation that private information is overweighted only if it contradicts the predecessor' belief.\footnote{\citet{dominitz:09} also collects subjects' belief data in a stylized information cascade experiment and finds that private information is not overweighted.} By contrast, we collect subjects' reservation price data and choice data and provide direct evidence that private information is valued more than social information.

Our paper is also related to the literature on non-Bayesian social learning models. Non-Bayesian social learning models, for the most part, are motivated by the casual observation that social networks are often prohibitively complex; consequently, bounded rational individuals may not be able to make Bayesian inferences on the basis of social information. Models that characterize their incorporation of social information in a certain non-Bayesian manner are then proposed.\footnote{See \citet{bala:98}, \citet{rabin:10}, \citet{guarino:13}, \citet{eyster:10,eyster:14}, and \citet{bohren:16}, among others.} Among these non-Bayesian models, most recently in \citet{molavi:18}, individuals are assumed to process private information in a Bayesian manner, which is treated as common knowledge among them. A class of non-Bayesian social learning models initiated by \citet{degroot:74}, which assumes a linear aggregation of private information and social information, also exists.\footnote{See \citet{demarzo:03}, \citet{golub:10,golub:12}, and \citet{jadbabaie:12}, among others.} By contrast, our belief error model focuses on a basic deviation from the first-order belief assumption and applies to social learning settings where the network structure is sufficiently simple and individuals are fully informed of others' information structure with a simple form. Therefore, our paper should be viewed as complementary to the existing work on non-Bayesian social learning models.

Moreover, our paper is related to the literature on non-Bayesian individual decision making. Individuals have been shown to deviate from the Bayesian updating paradigm systematically in many settings of individual decision making.\footnote{See, for example, \citet{tversky:74}, \citet{holt:09} and the surveys in \citet{camerer:95}, \citet{rabin:98}, and \citet{camerer:04}.} Alternative models that capture a certain deviation from Bayesian updating have been proposed via two typical approaches. In the first approach, non-Bayesian descriptive models are formalized to characterize individuals' misinterpretation of the signal generating process or signals \citep{barberis:98,rabin:99,rabin:02,rabin:10}. In the second approach, non-Bayesian decision models are built on axiomatic foundations and are applied to settings where the prior probability is subjectively determined and even adjusted in the presence of new observations \citep{epstein:06,epstein:08,ortoleva:12}. By contrast, our non-Bayesian decision model applies even when both the prior probability and the signal generating process are given objectively and can be reasonably agreed on by individuals. Moreover, our model can explain the experimental findings that are difficult to reconcile with the aforementioned non-Bayesian decision models, e.g., the monotonic relationship between signal quality and the frequency of making a Bayesian choice.\footnote{For example, \citealp{epstein:10} (in the spirit of \citealp{epstein:06} and \citealp{epstein:08}) assumes that agents' posterior belief is a linearly weighted sum of the Bayesian posterior probability and the prior probability.  This assumption always predicts a Bayesian first choice in our experiment, which is inconsistent with the data. }

Finally, our paper is related to the literature on stochastic binary choice models. A canonical model of stochastic binary choice is the logit model first proposed by \citet{bradley:52}.
The logit model typically assumes that individuals optimally make deterministic choices given their knowledge, but from the perspective of observers, they follow a logistic response function for choice probabilities. The gap has traditionally been reconciled by assuming that individuals evaluate the payoffs of choices with some noise or that an  individual-specific preference shock exists, and that observers are ignorant of them. In contrast to the payoff disturbance foundation, \citet{matejka:15} provides a rational inattention foundation for the logit model. They show that when individuals endogenously determine their learning about the payoffs of choices and the information cost takes the form of the Shannon mutual information cost function, their optimal information and choice strategy naturally implies the logistic choice rule.\footnote{\citet{caplin:15} also provides a rational inattention foundation for agents' stochastic choice without imposing specific assumptions about the function form of the information cost.} We argue in Section \ref{sec:belief-disturbance-model} that applying logit models with either foundation to our setting is inappropriate. Instead, we propose a novel belief error model to generate the stochastic binary choice. In terms of modeling error, our method of making assumptions on random posterior beliefs is different from the typical method of assuming an additively separable error term, as in the logit model.


The remaining parts of the paper proceed as follows. Section \ref{sec:bayesian-benchmark} presents a Bayesian benchmark model of our social learning game. Sections \ref{sec:experimental-design} and \ref{sec:experimental-results} report experimental design and results correspondingly. Section \ref{sec:belief-disturbance-model} proposes a model of belief error and structurally estimates the model parameters using maximum likelihood estimation. Finally, section \ref{sec:conclusion} concludes.

\section{Bayesian Benchmark Model}
\label{sec:bayesian-benchmark}

In the benchmark model, we assume that Bayesian agents share a common prior belief about the state of the world, i.e., state 1 and state 2 occur with probability $\frac{1}{2}$. Each agent has to make a binary action decision whose payoff is contingent on the underlying state. Specifically, action 1 delivers a payoff of 1 in state 1 and a payoff of 0 in state 2, and action 2 delivers a payoff of 0 in state 1 and a payoff of 1 in state 2.

Independent signals $\{s_n\}_{n \ge 0} $, where $ s_n \in  \{B, W \}$, provide incomplete information about the underlying state. We assume that the signal structure is symmetric, i.e.,
$P(s_n=B\mid\text{state 1}) = P(s_n=W\mid\text{state 2}) = p$ and $P(s_n=W\mid\text{state 1}) = P(s_n=B\mid\text{state 2}) = q \equiv 1-p $. We also assume that $p>\frac{1}{2}$.

In a private information setting, an agent is endowed with a private signal $s_0$.  Conditional on the realization of $s_0$, she makes a first choice between the two actions and may also decide to acquire a certain number of additional signals conditional on which she is entitled to make a second choice.  A social information setting is almost identical, except that the additional observations consist of a certain number of other agents' first choices, each of which is made after observing an independent realization of signal $s_0$.

Clearly, a Bayesian agent's posterior belief about state 1 is $p$ after observing a signal $s_0 =B$ and is $q$ after observing a signal $s_0 =W$. Bayesian agents who maximize expected utility should optimally choose action 1 in the former case and action 2 in the latter case. It then can be inferred from observing an agent's first choice of action 1 (action 2) in the social information setting that a signal with the realization of $B$ ($W$) is observed. Thus, the informational content contained in a signal is identical to that inferred by observing an agent's first choice. It is then straightforward that an agent's willingness to pay for observing $n$ additional signals in the private information setting (denoted as $W_{s_0}^{pri}(p,n)$) is equal to her willingness to pay for observing $n$ agents' first choices in the social information setting (denoted as $W_{s_0}^{soc}(p,n)$).

\begin{proposition}
For Bayesian agents who maximize expected utility, $W_{s_0}^{pri}(p,n) =W_{s_0}^{soc}(p,n)$.
\label{prop:zero-information-value}
\end{proposition}

Since the prior belief about either state is $\frac{1}{2}$ and the signal structure is symmetric,  $W_{s_0}^{pri}(p,n)$ remains the same regardless of the realization of signal $s_0$. For notational simplicity, we use $W(p,n)$ below to indicate $W_{s_0}^{pri}(p,n)$ or $W_{s_0}^{soc}(p,n)$.

 If agents are further assumed to be risk neutral, their willingness to pay for additional observations can be characterized in a closed form, specifically, a form that is monotone in the number of additional observations and concave in the informativeness of signals. We provide the proof of the proposition in the supplementary Appendix B.

\begin{proposition}
Assume that Bayesian agents maximize expected payoff. Then, the following properties hold.
\begin{enumerate}[nosep,label=(\roman*)]
	\item $W(p,1) = 0$ and $W(p,2n)=W(p,2n+1) = \sum_{k=n+1}^{2n}C_{2n}^k\left(p^kq^{2n+1-k}-p^{2n+1-k}q^k\right)$ for $n \geq 1$.
	\item  $W(p,2m) > W(p,2n)$ whenever $m>n$, and $\lim_{n\to\infty} W(p,2n) = 1-p$.
	\item  $\partial W(p,2n)/\partial p > 0$ for $p < p^*_{2n}$ and $\partial W(p,2n)/\partial p < 0$ for $p > p^*_{2n}$, where $p^*_{2n} = \frac{1}{2}+ \frac{1}{2}\left[1-\sqrt[n]{(2\cdot 4\cdot \cdots \cdot 2n)/(3\cdot 5\cdot \cdots \cdot (2n+1))}\right]^{\frac{1}{2}}$.
	\item The threshold level of informativeness of signals that makes additional observations most valuable to an agent, $p_{2n}^*$, decreases as the number of additional observations $2n$ rises, and $\lim_{n\to\infty} p^*_{2n} = \frac{1}{2}$.
	\item $\partial^2 W(p,2n)/\partial p^2 < 0$.
\end{enumerate}
\label{prop:bayesian-expected-payoff}
\end{proposition}

The property that there is no informational value of observing an additional signal holds regardless of risk attitude. As for observing more than one additional signal, one can show that the increment in expected utility must be proportional to $\sum_{k=n+1}^{2n}C_{2n}^k\left(p^kq^{2n+1-k}-p^{2n+1-k}q^k\right)$ regardless of risk attitude. Nevertheless, the assumption of risk neutrality is necessary for a closed-form characterization of willingness to pay because utility levels at four points are involved.\footnote{Utility levels at four points are the utility of payoff of one, the utility of payoff of zero, the utility of payoff of one subtracted by willingness to pay, and the utility of payoff of zero subtracted by willingness to pay. }

\section{Experimental design and hypotheses}
\label{sec:experimental-design}

We design an experiment to collect data on subjects' choices and their reservation prices for additional observations in both private information and social information settings. In each setting, we vary the signal quality and signal quantity parameters.  Specifically, signal quality ($p$) varies from 0.6 to 0.75 to 0.9, i.e., from low to medium to high accuracy. Signal quantity ($n$) takes a value of either 1 or 3 and determines whether one or three additional observations are provided. A total of twelve treatment conditions are considered, and we implement a within-subject design. An experimental session consists of sixty decision rounds, of which the first twelve correspond to the twelve treatment conditions, which are followed by four repetitions.\footnote{The display order of the twelve treatment conditions in our experiment is: (i) private information, $n = 1$, $p=0.6 \rightarrow 0.75 \rightarrow 0.9$; (ii) private information, $n = 3$, $p = 0.6 \rightarrow 0.75 \rightarrow 0.9$; (iii) social information, $n = 1$, $p = 0.6 \rightarrow 0.75 \rightarrow 0.9$; (iv) social information, $n = 3$, $p = 0.6 \rightarrow 0.75 \rightarrow 0.9$.}

In addition to learning from private information and learning from social information, another type of learning is inherent in laboratory experiments: participants' learning the experiment during the course of the experiment. Specifically, for a multi-round experiment, whether participants receive feedback about their performance at the end of each round may have an impact. To investigate the potential effect of participants' learning about the experiment through feedback, we implement a between-subject design in which a decision round consists of the following three stages in no-feedback sessions, and in feedback sessions there is also the fourth stage in which feedback is provided.

In stage 1, a computer randomly selects either of urn 1 and urn 2 for use in that decision round. Urn 1 of type $\frac{12}{20}$ contains twelve black balls and eight white balls, and urn 2 of this type contains twelve white balls and eight black balls. The compositions of type $\frac{15}{20}$ urns and type $\frac{18}{20}$ are similarly determined.
While the type of urns (i.e., $\frac{12}{20}$, $\frac{15}{20}$, or $\frac{18}{20}$) is revealed to all subjects, the label of the urn (i.e., urn 1 or urn 2) is not known. The computer independently and randomly draws one ball from the urn with replacement for each subject and informs her of the color of the drawn ball. On the basis of the color of the drawn ball, each subject is asked to make a first binary choice, i.e., guess which of the two urns is used.

In stage 2, each subject is endowed with 300 tokens and asked to state her reservation price for additional information by choosing a number from $\{0, 1, 2,...,299, 300 \}$. The chosen number, say $b$, refers to her willingness to pay $b$ tokens for additional information. Before stating her reservation price, the subject is informed of the composition of the additional information. Specifically, in the private information treatment condition with $n \in \{1,3\}$, the additional information includes the color(s) of $n$ ball(s) randomly drawn from the urn used in that round. In the social information treatment condition with $n \in \{1,3\}$, the additional information is the first choices of one or three other subjects.

In stage 3, the Becker-DeGroot-Marschak method is employed to determine whether subjects successfully purchase additional information. Specifically, the computer randomly and equally chooses an ask price from $\{0, 1, 2,...,299, 300 \}$, say $s$ tokens. If a subject's reported reservation price exceeds the ask price ($b \geq s$), the subject is charged $s$ tokens (collected from her endowment of 300 tokens) and is provided with the additional information. After the additional observations, the subject is asked to make a second binary choice between urn 1 and urn 2. In this case, her second choice becomes her final choice in that round. If a subject's reported reservation price is below the ask price ($b < s$), the subject is neither charged nor provided with additional information. In this case, her first choice is also her final choice in that round. In feedback sessions, a fourth stage is implemented after stage 3, in which subjects are told of the urn that is actually used in that round.

Subjects are paid according to their final choices in a randomly selected decision round. Specifically, the computer randomly selects one of the sixty decision rounds to serve as the paid round. We say that a subject makes a correct choice if her final choice in the paid round is the same as the urn that is actually used in that round and otherwise makes an incorrect choice. A subject earns 300 tokens for a correct choice and zero tokens for an incorrect choice. In addition, she retains her endowment of 300 tokens if she does not purchase additional information in the paid round. If she receives additional information at a cost of $s$ tokens in the paid round, she retains $300-s$ tokens of the endowment. Finally, the total number of tokens earned is redeemed for Chinese yuan at an exchange rate of $\frac{1}{10}$.

According to the Bayesian benchmark model in Section \ref{sec:bayesian-benchmark}, we have the following two experimental hypotheses.

\textbf{Hypothesis 1 } The first choice in a decision round is urn 1 (urn 2) if a black (white) ball is observed. The second choice (if any) is urn 1 if the majority of signals are in favor of urn 1 and urn 2 if the majority of signals are in favor of urn 2.\footnote{Let $B$ and $W$ refer to observing a black ball and a white ball, respectively, and let $C_1$ and $C_2$ refer to a choice of urn 1 and a choice of urn 2 in stage 1, respectively. Signals in favor of choosing urn 1 for the second time include $\BB$, $B\one$, $4B$, and $3B1W$ in the private information setting and $1B3\one$, $1B2\one1C_2$, and $1W3C_1$ in the social information setting. Signals in favor of choosing urn 2 for the second time include $\WW$, $WC_2$, $4W$, $3W1B$ in the private information setting and $1W3C_2$, $1W2C_21C_1$, and $1B3C_2$ in the social information setting.\label{signalmajority}}

\textbf{Hypothesis 2 } Participants' reservation prices for observing additional signals are identical in the private information and social information treatment conditions for any given signal quality and signal quantity. The reservation price in each treatment condition is shown in Table \ref{hypothesis}.

\begin{table}[htb]
\centering
\caption{The same value of private information and social information (unit: tokens)}
\label{hypothesis}
\begin{tabular}{cccc}
\toprule
& $p=0.6$ & $p=0.75$ & $p=0.9$\\
\midrule
$n=1$ &    0                &   0      &  0  \\
$n=3$ &     14.4               &   28.1        &  21.6  \\
\bottomrule
\end{tabular}
\end{table}

We recruited 100 undergraduate students from Xiamen University as participants by sending email invitations to a subset of candidate subjects registered in our subject pool. The study was conducted at the Finance and Economics Experimental Laboratory (FEEL) of Xiamen University in 2014. We ran five sessions with twenty subjects in each session, including three no-feedback sessions and two feedback sessions. Before the decision task, subjects received a copy of the experimental instructions and we read the experimental instructions aloud and answered any questions from the subjects (See the supplementary Appendix C for the Experimental Instructions). Each session of the study lasted approximately 50 minutes, and subjects earned an average of approximately 50 Chinese yuan (approximately 8 US dollars) and received their payments in private.

\section{Experimental Results}
\label{sec:experimental-results}

An observation of the experiment includes the parameters treatment condition (i.e., private or social information; $p=0.6$, $0.75$, or $0.9$; $n=1$ or $3$), a dummy variable indicating whether feedback was provided, the order of the decision round, the subject's experimental ID number, the label of the urn that was used in the decision round, signals that the subject observed, the subject's first choice and second choice, if any, and the subject's reported reservation price for observing additional signals. We have a total of 6000 observations in our data sample, and subjects observed the additional signals and made a second choice in 1230 observations.

Overall, in subsections \ref{subsec:experimental-result1} and \ref{subsec:experimental-result23} we present three main findings from the choice data and reservation price data at the aggregate level: (1) subjects value social information less than private information; (2) the frequency of subjects' making a Bayesian first choice increases as the signal quality increases; and (3) subjects report considerably higher reservation prices than predicted by the Bayesian benchmark model, specifically, there is positive informational value of observing one additional signal. Notably, data from the feedback sessions convey a similar message in terms of all three aspects as data from the no-feedback sessions, but the strength of the message from the former is slightly weaker. We also illustrate subjects' heterogeneity in terms of different valuations of the two types of information in subsection \ref{sec:individual-level-analysis}.

\subsection{Different Valuations of Private Information and Social Information}
\label{subsec:experimental-result1}
We find both direct and indirect evidence that subjects treat private information and social information differently, which refutes experimental hypothesis 2.

 Table \ref{meanbid} shows that subjects' average reservation prices are consistently much higher than the theoretical values based on the Bayesian benchmark model; more importantly, social information has a lower value than private information. We run the two-group mean test under the null hypothesis that reservation prices are identical in both the private and the social information settings with the alternative hypothesis that the average reservation price in the social information setting is lower than that in the private information setting. The results show that the null hypothesis is rejected in all conditions of the no-feedback sessions and is rejected in all but one condition of the feedback sessions.\footnote{An analysis of later rounds of a session, i.e., rounds 37-60, produces similar results.}

\begin{table}[htbp]
\centering
\caption{Average reservation prices in the private and social information settings}
\label{meanbid}
\resizebox{.7\hsize}{!}{
\begin{tabular}{ccccc}
\toprule
 & & $p=0.6$ & $p=0.75$ & $p=0.9$\\
\midrule
 \multirow{3}*{no feedback, $n=1$ }     & private   & 49 (3.3)     & 51 (3.2)   & 66 (4.6) \\
							  & social  & 40 (2.9)     & 42 (3.0)  & 51 (4.1) \\
							 &   & $0.0186$  & $0.0250$ & $0.0069$ \\[1.5ex]
 \multirow{3}*{feedback, $n=1$ }  & private & 54 (7.1)     & 37 (4.0)   & 28 (3.4) \\												
				                 & social   & 44 (5.8)     & 34 (4.2)  &23 (2.4) \\
					& 	& 0.2283  & $0.0388$ & $0.0254$  \\
\midrule
\multirow{3}*{no feedback, $n=3$ }  &private     & 94 (4.5)      & 97 (4.8) & 98 (5.5)  \\
						&social	& 74 (4.4)     &  76 (4.4) &  79 (5.4) \\
						&	& $0.0006$ &  $0.0005$ & $0.0087$ \\[1.5ex]
\multirow{3}*{feedback, $n=3$}     & private  & 83 (8.1)      & 66 (5.5) & 35 (3.6)  \\
						 & social  & 62 (5.7)     &  53 (5.0) &  32 (3.1) \\
						 &  & $0.0017$ & $0.0022$ & $0.0426$ \\
\bottomrule
\end{tabular}
}
\medskip
\begin{minipage}{\linewidth}\small
Notes: See standard errors in the parentheses and $p$-values of one-tailed tests of the corresponding mean comparison in the corresponding third row.
\end{minipage}
\end{table}

Figure \ref{fig:empirical-bid-cdf} further illustrates the difference in the empirical distributions of subjects' reservation prices between the private information setting and the social information setting. We run a Kolmogorov-Smirnov test of the equality of the distributions in the two settings for each of the six treatment conditions. The null hypothesis is rejected at the 5\% significance level in all six cases.\footnote{We also test the difference in the empirical distributions of the reported reservation prices between the two information settings for feedback session data and for no-feedback session data, separately. We run a Kolmogorov-Smirnov test of the equality of the distributions in the two settings for the two subsamples of data, and the null hypothesis is rejected at the 1\% significance level for each of the two subsamples. When we further divide each subsample of data into six treatment conditions and run a similar statistical test, we find that the null hypothesis is rejected at the 5\% significance level in each treatment condition of the subsample of no-feedback sessions and that the null hypothesis cannot be rejected at the 5\% significance level in each treatment condition of the subsample of feedback sessions. These results suggests that subjects may learn from feedback, which in turn has an impact on their valuation of the two types of information.
}

\begin{figure}[htbp]
	\centering
	\includegraphics[width=\textwidth, trim=0 10 0 20]{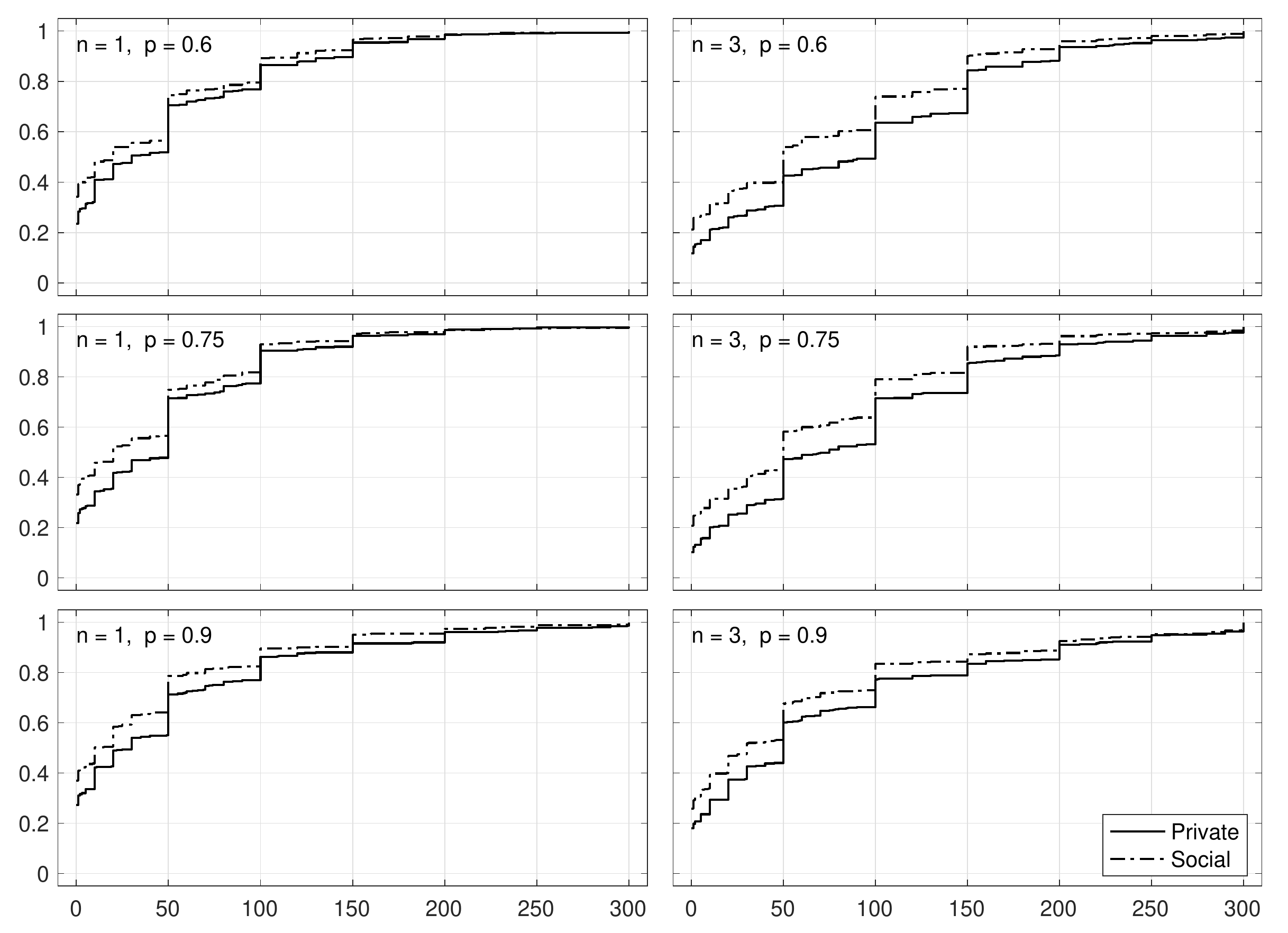} 
	\caption{Empirical distribution of reservation prices}
	\label{fig:empirical-bid-cdf}
\end{figure}

In the subsample of observations where subjects were asked to make a second choice, subjects' final choices were made on the basis of the initial information (the color of the ball that was shown in stage 1) and the new information, which includes either the colors of the additional balls or other subjects' first choices. If the impact of the new information on the second choice differs depending on the type of information, then subjects have a different valuation of private information and social information. Our indirect evidence is based on the finding that the effect is stronger when the new information is private information than when it is social information.

Specifically, we investigate in the subsample the percentage of decision rounds in which a subject's second choice contradicts the color of the ball in stage 1, for example, when the subject's second choice was urn 2 after observing a black ball in stage 1 or her second choice was urn 1 after observing a white ball in stage 1. We may interpret that the larger the percentage is, the bigger is the impact the new information has on subjects' second choices, everything else being equal. We find that, overall, the percentage is 22\%  in the private information setting and  16\% in the social information setting. We run a two-group mean test under the null hypothesis that the percentage is the same in both settings and find that the hypothesis is rejected at the 10\% significance level.\footnote{The overall percentage ignores the details of different realizations of the initial information and the new information. To address this question, we introduce an explanatory variable $z$ that counts the number of signals in the new information in favor of the initial information. For example, the variable takes a value of 1 if the new information contains exactly one signal in favor of the initial information, a value of 2 if the new information contains two signals in favor of the initial information, and so on. When a tie occurs between signals in favor of either choice, i.e., $z=0$ when $n=1$ or $z=1$ when $n=3$, the percentage is higher in the private information setting than in the social information setting, indicating a stronger effect of private information.}

\subsection{Signal Quality and Bayesian Choice}
\label{subsec:experimental-result23}
Since urn 1 consistently contains more black balls than white balls and urn 2 consistently includes more white balls than black balls and because the prior belief about either urn is $\frac{1}{2}$, we say that a subject's first choice is a Bayesian choice if she chooses urn 1 after observing a black ball or chooses urn 2 after observing a white ball. Table \ref{bayesianchoice} reports the percentage of Bayesian choices in the first stage. While it is not surprising that the percentage is less than 100\%, i.e., $94\%$, it is interesting that the percentage increases as the signal quality increases.\footnote{An analysis of rounds 37-60 of the feedback sessions produces similar results.}\footnote{\citet{noth:03} extends \citet{anderson:97}'s sequential social learning experimental design by introducing two signal qualities, i.e., $p=0.6$ and $p=0.8$, and finds that the frequency of Bayesian choice in the first position is $85.8\%$ when $p=0.6$ and $97\%$ when $p=0.8$. Since subjects in their experiment always learned the true state at the end of each round, they argue that subjects' performances in the previous round(s), such as gambler's fallacy, can explain a portion of the non-Bayesian choices in the first position. By contrast, subjects in our experimental sessions with no feedback cannot learn their performance in the previous round(s), which precludes the above possible explanation.}\footnote{Similarly, we say that a subject's second choice (if any) is a Bayesian choice if she chooses the urn favored by the majority of signals, as specified in Footnote \ref{signalmajority}. Among the 636 observations in which the majority of signals are in favor of either urn 1 or urn 2, subjects made a Bayesian choice in 612 cases, which amounts to a Bayesian choice percentage of $96\%$.}

\begin{table}[htbp]
\centering
\caption{The effect of signal quality on Bayesian first choice}
\label{bayesianchoice}
\resizebox{.93\hsize}{!}{
\begin{adjustbox}{center}
\begin{tabular}{ccccccc}
\toprule
 & $p=0.6$ & $p=0.75$ & $p=0.9$ & $p=0.6$ vs. 0.75 & $p=0.75$ vs. 0.9 & $p=0.6$ vs. 0.9\\
\midrule
no feedback       & $87\%$     & $96\%$   & $97\%$  & $0.0000$  & $0.0380$ & $0.0000$ \\[1ex]
feedback       & $90\%$      & $97\%$ & $99.5\%$  & $0.0000$ & $0.0001$ & $0.0000$\\
\bottomrule
\end{tabular}
\end{adjustbox}
}

\medskip
\begin{minipage}{\linewidth}\small
Notes: The percentage refers to the extent to which the first choice is a Bayesian choice. The last three columns report $p$-values of two-tailed tests of the corresponding mean comparison.
\end{minipage}
\end{table}

\subsection{Individual-level analysis}
\label{sec:individual-level-analysis}

At the individual level, we find considerable heterogeneity across subjects. Figure \ref{fig:heterogeneity} illustrates each subject's average reservation price for additional signals in the private information and social information settings for each condition of signal quality and signal quantity. Some subjects value private information more, some subjects value social information more, and others value the two types of information more or less equally.

\begin{figure}[htbp]
	\centering
	\includegraphics[width=0.94\textwidth, trim=0 20 0 20]{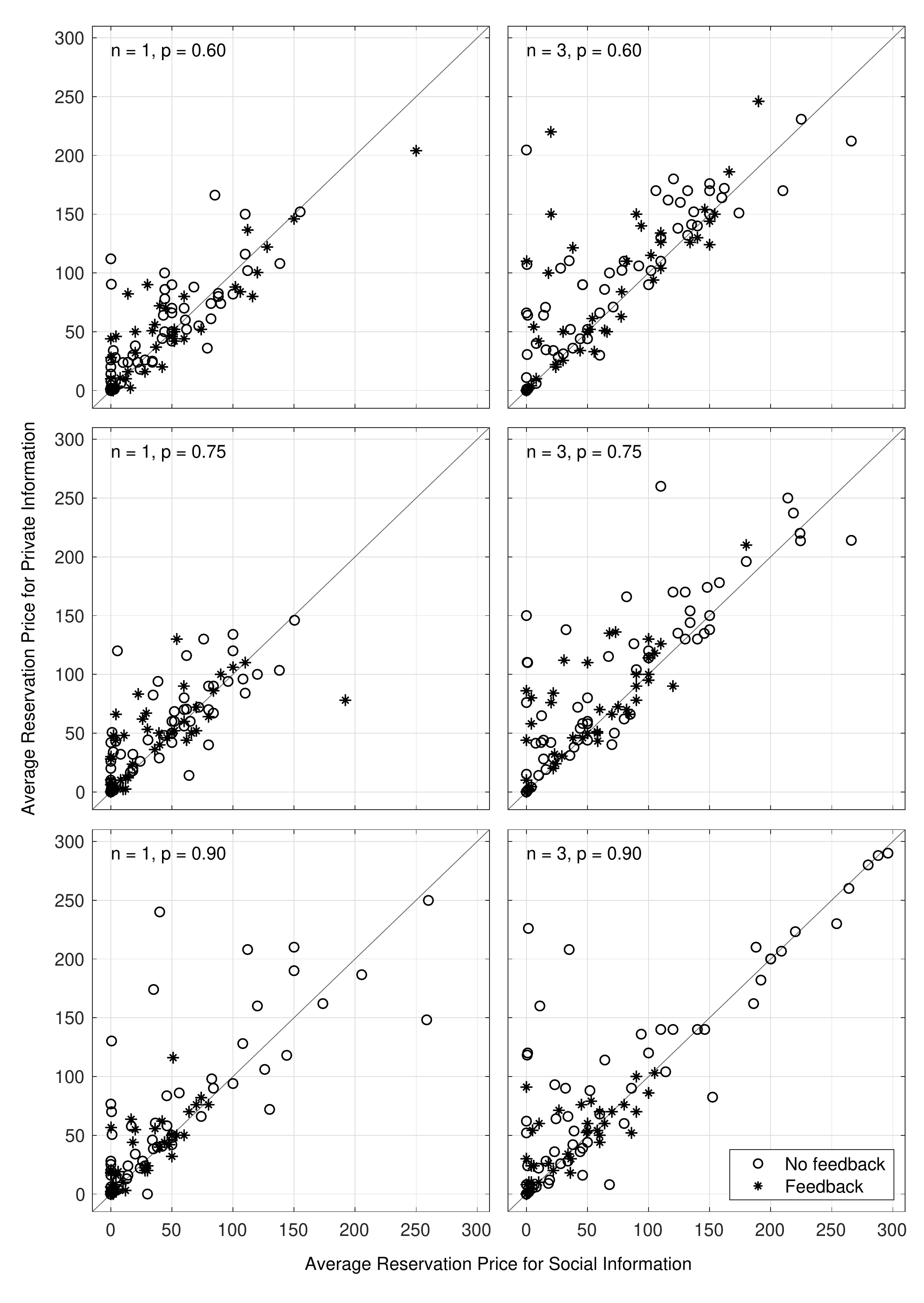} 
	\caption{Individual average reservation prices for private and social information}
	\label{fig:heterogeneity}

    \medskip
    \begin{minipage}{\linewidth}\small
	Note: The diagonal line refers to the situation where reservation prices for the two types of information are equal.
    \end{minipage}
\end{figure}

\section{Belief Error Model and Structural Estimation}
\label{sec:belief-disturbance-model}

In this section, we provide a belief error foundation for participants' systematic non-Bayesian choice behavior and their different valuations of private information and social information. A structural estimation of the model is then employed to gain deeper insights into the experimental data.

We are aware that the toolbox includes the classical logit choice model, which has the potential to explain our experimental findings. Nevertheless, we argue that an application of such a model to our experimental setting is not appropriate for the following reasons. When applying the logit model to generate a stochastic binary choice, the literature assumes either that individuals follow a logistic response function to make stochastic choices or that individuals optimally make a deterministic choice given their knowledge but from the perspective of observers, they follow a logistic response function for their choice probabilities. In the former case, individuals do not choose optimally (e.g., \citealt{kubler:04}). The model predicts that subjects report a higher reservation price for additional signals when they do not make a Bayesian first choice than when they make a Bayesian first choice.\footnote{The intuition is that after not making a Bayesian first choice subjects realize that they made a choice mistake and that the expected payoff is lower. Henceforth, the benefit of observing additional signals larger given a fixed ex ante expected payoff of making a second choice.} This prediction is not consistent with our experimental finding that subjects' reported reservation prices in the two cases are not significantly different ($p$-value $= 0.1885$; two-group mean-comparison test). In the latter case, the logit model has a payoff disturbance foundation or a rational inattention foundation. The model with a payoff disturbance foundation assumes that individuals have payoff disturbances or preference shocks, of which observers are ignorant. An application of the payoff disturbance model (e.g., logit QRE, see \citealt{goeree:07}) represents a conceptual challenge because our experimental data include reservation price data in addition to choice data. Since a subject's reservation price in principle reflects the difference in her expected payoffs before and after observing additional signals, explaining the variation in reservation price by the assumption of payoff disturbance leads to concerns of circular reasoning.\footnote{We also employ the payoff disturbance model to structurally estimate the experimental data by pretending to ignore the conceptual concern and find that it performs worse than our method, as illustrated in Section \ref{estimationresult}.} According to the model with a rational inattention foundation, individuals are assumed to choose their information acquisition strategy optimally when they endogenously determine their way of acquiring information about actions' payoffs, and the information acquisition cost is assumed to take the form of the Shannon cost function (e.g., \citealt{matejka:15}). Under this assumption, individuals' posterior beliefs are optimally chosen, and action choices based on these posterior beliefs follow a logistic response function form. Differently in our experimental setting, information acquisition must take the form of drawing independent signals with an exogenously given signal structure, and after receiving the signals, there is minimal room for unobservables that may justify individuals' choosing their information strategy. Thus, it is not compelling to justify the application of the logit model with a rational inattention foundation in our setting.

\subsection{A Basic Model with Belief Error}

The basic belief error model modifies the Bayesian benchmark model in the following three aspects.  First, subjects form random posterior beliefs with a Bayesian updating kernel and have a sophisticated consideration of their own random posterior beliefs. Second, subjects make a sophisticated consideration of others' random posterior beliefs when interpreting others' action decisions. Third, subjects also gain utility from the event of making a correct choice ex post in addition to the monetary payoffs realized from making a correct choice. Clearly, the generality and applicability of the three assumptions is ranked in descending order.

\subsubsection{Random posterior belief with a Bayesian kernel}
Subjects sometimes systematically deviate from Bayesian updating even in many simple settings of individual decision making and often exhibit behavioral biases (e.g., \citealt{tversky:74}; \citealt{holt:09}). Recently, \citet{defilippis:17} collect subjects' belief data and find that while a high percentage of beliefs are in line with Bayesian updating, a considerable percentage are smaller or greater than the Bayesian ones. In addition, a small proportion of beliefs are even in the opposite direction specified by Bayesian updating.

We model the variation from Bayes' law by assuming that when updating beliefs from prior beliefs and observations, subjects form a random posterior belief dependent on  the Bayesian posterior probability, i.e., the posterior probability calculated according to Bayes' law. Specifically, let $p_x$ be the Bayesian posterior probability of urn 1 given a signal set $x$. Subject $i$'s posterior probability of urn 1 takes the form of a random variable $\tilde{p}_x^i$; correspondingly, her posterior probability of urn 2 is $1- \tilde{p}_x^i$. We assume that the value of the random posterior probability is realized only after observing signal set $x$ and that it is privately known to subject $i$. We also assume that sophisticated subjects
take the randomness of posterior belief into account when extrapolating posterior beliefs in any future information set. We may infer that the randomness comes from belief errors when the subject forms her posterior probability.\footnote{ The form of $\tilde{p}_x^i$ implicitly assumes that in an information set including the first signal and additional signals, the subject forms her posterior belief by taking a new draw of noise in her belief without taking into account the old draw of noise in her belief after observing the first signal. With this simplification, the potential effect of the realized noise in belief after observing the first signal on the noise in belief after observing the first signal and additional signals is ignored.}

Since the posterior probability $\tilde{p}_x^i$ is a random variable that takes a value from 0 to 1, an ideal probability distribution of the random variable should have support on the interval $[0,1]$. Among common continuous distribution functions with support on a bounded interval, we find that the beta distribution is particularly appealing since it can successfully capture a few intuitions underlying the belief error model and generate results with closed-form expressions.\footnote{Alternatively, one may define subject $i$'s posterior belief to be the Bayesian posterior probability plus an error term, i.e., $\tilde{p}_x^i =p_x + \epsilon _x^i$. This definition will inevitably result in a truncated distribution if we assume that $\epsilon _x^i$ follows a distribution with unbounded support such as normal distribution. While the truncation problem could be resolved by applying the logit transformation twice, e.g., $\ln \frac{\tilde{p}^i_x}{1-\tilde{p}_x^i} = \ln \frac{p_x}{1-p_x} + \epsilon_x^i$, the appealing feature of generating closed-form predictions will be lost. \label{separableerror}}

\begin{ass}
	\label{ass:distribution-p}
	Subject $i$'s posterior belief $\tilde{p}_x^i$ follows a beta distribution with two parameters $\gamma^i \ge 0$ and $p_x \in [0,1]$. Specifically,
	\begin{equation*}
	\tilde{p}_x^i \sim \operatorname{Beta}\left(\frac{p_x}{\gamma^i}, \frac{1-p_x}{\gamma^i}\right) 
	\end{equation*}
	when $\gamma^i >0$ and $p_x \in (0,1)$. $\tilde{p}_x^{i} =p_x$ almost surely when $\gamma^i =0$, $\tilde{p}_x^{i} =0$ almost surely when $ p_x =0$, and $\tilde{p}_x^{i} =1$ almost surely when $p_x =1$.
\end{ass}

Assumption \ref{ass:distribution-p} has a few desirable features. First, it guarantees that $\tilde{p}_x^i$, as a probability, is always bounded by $0 \leq \tilde{p}_x^i \leq 1$ since the support of the beta distribution is $[0,1]$. Second, the expectation $\mathbb{E}(\tilde{p}_x^i) = p_x$, which suggests that the subject's belief is on average ``right'' in the sense that the random posterior belief overall has a Bayesian posterior belief kernel. Third, the variance $\var(\tilde{p}_x^i) = \frac{\gamma^i}{1+\gamma^i} p_x(1-p_x) $ shrinks as the parameter $\gamma^i$ decreases and converges to zero as $\gamma^i$ approaches zero. Therefore, the parameter $\gamma^i$ can naturally be interpreted as the degree of subject $i$'s belief error: the smaller the parameter is, the smaller her belief error is. $\gamma^i = 0$ corresponds to a special case in which the distribution of $\tilde{p}_x^i$ is degenerate, or equivalently, $\tilde{p}_x^i = p_x$ almost surely. Finally,  for any fixed $\gamma^i$, the variance is proportional to $p_x (1-p_x)$, which has the highest value when $p_x =\frac{1}{2}$. The variance is also symmetric around $\frac{1}{2}$ and decreases when $p_x$ changes in the direction of 0 or 1. In the boundary cases of $p_x = 0$ or $1$, $\tilde{p}_x^i = 0$ almost surely or $\tilde{p}_x^i = 1$ almost surely. This feature captures the intuition that a subject's belief error shrinks as the corresponding Bayesian posterior belief approaches the boundaries.

\citet{nyarko:06} use a beta distribution to characterize the posterior belief in a very different setting. In their setting, subjects are interested in forming a belief about the probability of a profitable state, and the prior belief is that the probability follows a beta distribution. After observing occurrences of profitable and unprofitable states, which follow a Bernoulli distribution, they apply Bayes' law to form a posterior belief about the probability, which must be a beta distribution because a beta prior distribution is a conjugate prior for the Bernoulli likelihood function. By contrast, subjects in our setting are interested in forming a belief about the underlying binary state, and the prior belief is a fixed number. After observing signals conditional on the underlying binary state, which follows a Bernoulli distribution, subjects' posterior beliefs are a fixed number if Bayes' law is applied. It is only in our non-Bayesian paradigm that subjects' posterior beliefs are assumed to be random and follow a beta distribution. In summary, our paper is essentially different from theirs in terms of the information structure, the contents of the prior and posterior beliefs, and the Bayesian/non-Bayesian paradigm.

When maximizing expected utility, subject $i$'s optimal choice strategy according to the belief error model is similar to that of the Bayesian benchmark model: choose urn 1 if $\tilde{p}_x^i > \frac{1}{2}$, choose urn 2 if $\tilde{p}_x^i < \frac{1}{2}$, and be indifferent between the two options if $\tilde{p}_x^i = \frac{1}{2}$. Let $P(\text{correct choice}_i \mid x,\tilde{p}_x^i)$ be the (subjective) probability of making a correct choice after subject $i$ observes signals $x$ and forms a realized value of her posterior belief $\tilde{p}_x^i$. The optimal choice strategy implies that $P(\text{correct choice}_i \mid x,\tilde{p}_x^i) = \max(\tilde{p}_x^i, 1-\tilde{p}_x^i)$. Then, the probability of making a correct choice for information set $x$ before the value of $\tilde{p}_x^i$ is realized is
$$P(\text{correct choice}_i \mid x) = \mathbb{E}P(\text{correct choice}_i \mid x,\tilde{p}_x^i) = \mathbb{E}\max(\tilde{p}_x^i, 1-\tilde{p}_x^i) ,$$
which reflects the subject's assessment of her chance of making an ex post correct choice for a hypothetical information set $x$. For comparison, the corresponding assessment according to the Bayesian benchmark model is $\max(p_x, 1-p_x)$. We show that the two assessments prescribe the same ranking of information sets.

\begin{proposition}
\label{prop:subjective-assessment}
	Assume that subjects maximize expected utility. Under Assumption \ref{ass:distribution-p}, for signal sets $x$ and $y$,
	\begin{enumerate}[nosep, label=(\roman*)]
		\item if $\max(p_x, 1-p_x) = \max(p_y, 1-p_y)$, i.e., $p_x = p_y$ or $p_x=1-p_y$, then $P(\text{correct choice}_i \mid x) = P(\text{correct choice}_i \mid y)$;
		\item if $\max(p_x, 1-p_x) > \max(p_y, 1-p_y)$, then $P(\text{correct choice}_i \mid x) > P(\text{correct choice}_i \mid y)$.
	\end{enumerate}
\end{proposition}
\begin{proof}
See Appendix \ref{pf:subjective-assessment}.
\end{proof}

Since the assessment depends on $p_x$, we denote $v(p_x) \equiv P(\text{correct choice}_i \mid x) $ after skipping the index of subject $i$ and belief error measure $\gamma$. The second property of Proposition \ref{prop:subjective-assessment} implies that when $p_x > \frac{1}{2}$, the assessment according to the belief error model is increasing in $p_x$. The monotonicity of the assessment is trivially satisfied in the Bayesian benchmark model because $\max(p_x, 1-p_x) =p_x$ when $p_x > \frac{1}{2}$. Another interesting observation is that when $p_x > \frac{1}{2}$, the assessment according to the belief error model increases at an increasing rate, i.e., $\partial ^2 v(p_x)/\partial p_x^2 > 0$, whereas the assessment according to the Bayesian benchmark increases at a linear rate. We do not have a formal proof of the convexity of $v(p_x)$ due to the intricate interaction between the function form and the beta distribution, but our simulation exercise suggests that $\partial ^2 v(p_x)/\partial p_x^2 > 0$ for a wide range of values of belief error measure $\gamma$.

 It is straightforward to check that a subject who chooses optimally according to the belief error model also makes a Bayesian choice if and only if
$\left(\tilde{p}_x^i-\frac{1}{2}\right)\left(p_x -\frac{1}{2}\right)>0$.\footnote{We skip the tie case $\tilde{p}_x^i =\frac{1}{2}$ since belief error is continuous according to Assumption \ref{ass:distribution-p}.}
So the chance of subject $i$'s making a Bayesian choice is $P\left(\left(\tilde{p}_x^i-\frac{1}{2}\right)\left(p_x -\frac{1}{2}\right)>0\right)$. Proposition \ref{prop:bayesian-monotone} below shows that the chance is increasing in $p_x$ when $p_x > \frac{1}{2}$ and decreasing in $p_x$ when $p_x < \frac{1}{2}$. Since for a given prior belief, the Bayesian posterior probability $p_x$ is increasing in the signal quality $p$ when $p_x > \frac{1}{2}$ and decreasing in $p$ when $p_x < \frac{1}{2}$, the proposition implies that the chance of subjects' making a Bayesian choice increases as the signal quality increases, which is consistent with the experimental finding that the frequency of making a Bayesian first choice increases from $p=0.6$ to $0.75$ to $0.9$. For comparison, the Bayesian benchmark model predicts that the chance of making a Bayesian first choice is always 100\%, regardless of the signal quality $p$.

\begin{proposition}
\label{prop:bayesian-monotone}
	Assume that subjects maximize expected utility. Under Assumption \ref{ass:distribution-p}, for signal sets $x$ and $y$, if $\left|p_x -\frac{1}{2}\right| > \left|p_y - \frac{1}{2}\right|$, then \[ P\left(\left(\tilde{p}_x^i-\frac{1}{2}\right)\left(p_x -\frac{1}{2}\right)>0\right) > P\left(\left(\tilde{p}_y^i-\frac{1}{2}\right)\left(p_y -\frac{1}{2}\right)>0\right). \]
\end{proposition}
\begin{proof}
See Appendix \ref{pf:bayesian-monotone}.
\end{proof}

We now show that the belief error model predicts a positive informational value of observing one additional signal. The intuition proceeds as follows. A subject's posterior belief after observing the first private signal in the first stage may not be equal to the Bayesian posterior probability, and a sophisticated subject extrapolates that she will form a random posterior belief for each subsequent information set in later stages. For any subsequent information set in later stages, she will choose optimally contingent on the realization of the random posterior probability, which is a better strategy than always choosing a certain action for all subsequent information sets. Thus, the belief error in the late stage universally increases the value of observing one additional signal. However, the effect of the belief error in the first stage is ambiguous and essentially depends on the value of the realized posterior in the first stage. Specifically, when the realized posterior is close to $\frac{1}{2}$, i.e., she is not confident about the true state, the second effect is also positive and the overall effect of observing one additional signal is positive. When the realized posterior is close to 0 or 1, i.e., she is confident about the true state, the second effect is negative and even dominates the first effect, so the overall effect is negative.

Formally, let $V_{B+1}^{i,pri}$ be subject $i$'s expected payoff of observing one additional signal after having already observed a black ball in the first stage. We simplify the notation as $V_{B+1}^{pri}$ whenever it is clear.\footnote{The assumption that the first observed ball is black has no loss of generality, and we keep the simplified assumption henceforth. Subject $i$'s expected payoff after choosing optimally in the first stage is $\max(\tilde{p}_B^i, 1-\tilde{p}_B^i)$ and $V_{B+1}^{pri} = \tilde{p}_B^i \cdot [p \cdot \mathbb{E}\max(\tilde{p}_{BB}^i, 1-\tilde{p}_{BB}^i) + q \cdot \mathbb{E}\max(\tilde{p}_{BW}^i, 1-\tilde{p}_{BW}^i)]
+ (1-\tilde{p}_B^i)\cdot
[p\cdot \mathbb{E}\max(\tilde{p}_{BW}^i, 1-\tilde{p}_{BW}^i) + q \cdot \mathbb{E}\max(\tilde{p}_{BB}^i, 1-\tilde{p}_{BB}^i)].$}


\begin{proposition}
\label{pf:positive-information-value-predict}
Assume that sophisticated subjects with belief error maximize expected payoff. Then under Assumption \ref{ass:distribution-p},
\begin{align*}
    V_{B+1}^{pri} - \max(\tilde{p}_B^i, 1-\tilde{p}_B^i)
    & \begin{cases} >0 & \text{if } \tilde{p}_B^i \in (\ubar{p},\bar{p}),\\
    \leq 0 &  \text{if } \tilde{p}_B^i \leq \ubar{p}   \text{ or } \tilde{p}_B^i \geq \bar{p},
    \end{cases}
\end{align*}
where $\ubar{p}, \bar{p}$ satisfy that $\ubar{p} < \frac{1}{2}<p_B< \bar{p}$ and $\bar{p}$ is increasing in $p_B$.
\end{proposition}
\begin{proof}
See Appendix \ref{pf:positive-private-info-value}.
\end{proof}

Proposition \ref{pf:positive-information-value-predict} suggests that when the realized posterior belief in the first stage is in the same direction as the Bayesian posterior belief and does not hit the boundary of high confidence about the true state, i.e., $\frac{1}{2} < \tilde{p}_B^i < \bar{p} $, there is positive value of observing one additional signal. For example, when the realized posterior belief is equal to the Bayesian posterior belief (i.e., $\tilde{p}_B^i= p_B$), observing one additional signal has positive informational value.

Furthermore, the difference in expected payoffs is generally positive from the perspective of an outside observer. Let $\mathbb{E}V_{B+1}^{pri}$ be a subject's expected payoff,  from the perspective of an observer,  of observing one additional signal after having already observed a black ball. Since the posterior belief is unknown to the observer, $\mathbb{E}V_{B+1}^{pri} = (p^2 +q^2)v(p_{BB}) + 2pqv(p_{BW})$
given that $\mathbb{E}\tilde{p}_B^i =p_B=p$. Similarly,  $\mathbb{E}\max(\tilde{p}_B^i, 1-\tilde{p}_B^i) =v(p_B)$ is the expected payoff in the first stage from the perspective of an observer. Since $(p^2 +q^2)p_{BB} + 2pq p_{BW} =p_B$, it is clear that $\mathbb{E}V_{B+1}^{pri} > v(p_B) $ if $v(p_x)$ is convex.

Finally, the belief error model also predicts that the value of observing three additional signals is always greater than the value of observing one additional signal, which is consistent with the experimental finding. We leave the formalization and the proof of this claim in the supplementary Appendix B.

\subsubsection{First-order belief about others' random posterior beliefs}
In the later stages of our social information setting, all signals except the first consist of other subjects' first choices, which are based on the colors of their own first balls drawn from the same urn. The question then is how a subject interprets the signals of others' first choices. In other words, what is the quality of such signals or, equivalently, what are the probabilities of observing another subject's first choice of urn 1 or urn 2 given the true urn state from the perspective of the subject. We provide a formal analysis below and demonstrate that how a subject interprets the signals of others' first choices depends on her belief about others' decision strategy and, in particular, on her belief about how others form their posterior beliefs.

Formally, the quality of a signal of another subject's first choice is characterized by $P^{i}(C_j\mid \text{urn $k$})$ ($j=1,2$ and $k=1,2$), which reflects subject $i$'s belief about another subject's probability of choosing urn $j$ given the true state urn $k$. Subject $i$'s belief about another subject's decision strategy is characterized by $P^{i}(C_j \mid B)$ and $P^{i}(C_j \mid W)$, which reflects subject $i$'s belief about another subject's probabilities of choosing urns $j=1,2$ after the other subject observes a black ball and a white ball, respectively. In contrast to the private information setting, where the quality of a signal of the color of a drawn ball is exogenously given and is naturally agreed on among subjects, the quality of a signal of another subject's first choice in the social information setting may be interpreted differently and is endogenously determined as follows.
\begin{align*}
	P^{i}(C_1\mid\urnone) &= P(B\mid\urnone)\cdot P^{i}(C_1 \mid B) + P(W\mid\urnone)\cdot P^{i}(C_1 \mid W) \\
	&= p \cdot P^{i}(C_1 \mid B)+q \cdot P^{i}(C_1 \mid W), \\
	P^{i}(C_2\mid\urntwo) &= P(W\mid\urntwo)\cdot P^{i}(C_2 \mid W) + P(B\mid\urntwo)\cdot P^{i}(C_2 \mid B) \\
	&= p \cdot P^{i}(C_2 \mid W)+q \cdot P^{i}(C_2 \mid B).
\end{align*}

If a subject believes that other subjects form posterior beliefs in a Bayesian manner and choose optimally given their beliefs, then $P^{i}(C_1\mid B) = P^{i}(C_2\mid W)=1$, which implies that $P^{i}(C_1\mid\urnone)=P(B\mid\urnone)$ and $P^{i}(C_2\mid\urntwo)=P(W\mid\urntwo)$. In this case, subject $i$ views the signal quality as identical regardless of whether it is a signal of another subject's first choice or it is a signal of the color of a drawn ball. Thus, she interprets a signal of another subject's choosing urn 1 the same as a signal of a black ball.

If a subject believes that other subjects form posterior beliefs in a non-Bayesian paradigm and choose optimally given their beliefs, then she may interpret that a signal of another subject's choosing urn 1 comes from observing a black ball or observing a white ball. In this case, the signal quality is generally viewed as different, i.e., $P^{i}(C_1\mid\text{urn 1})\neq P(B\mid\text{urn 1})$ and $P^{i}(C_2\mid\text{urn 2}) \neq P(W\mid\text{urn 2})$. In addition, knowing that $P^{i}(C_1\mid\urnone) > P(B\mid\urnone)$ if and only if $ P^{i}(C_1\mid W) / P^{i}(C_2\mid B) > p/q > 1$ and that $P^{i}(C_2\mid\urntwo) > P(W\mid\urntwo)$ if and only if $P^{i}(C_1\mid W) / P^{i}(C_2\mid B) < q/p < 1$, it is impossible that both $P^{i}(C_1\mid\text{urn 1}) > P(B\mid\text{urn 1})$ and $P^{i}(C_2\mid\text{urn 2}) > P(W\mid\text{urn 2})$. In other words, for subject $i$'s arbitrary modeling of others' forming posterior beliefs, the quality of a signal of another subject's first choice can never be uniformly improved compared to the signal quality in the private information setting.

We now know that how subject $i$ interprets a signal of another subject's first choice depends on her belief about the subject's decision strategy. This belief in turn depends only on her belief about the other subject's method of forming posterior beliefs, as long as the assumption that subjects choose optimally given their beliefs is maintained. 
Similar to our method of modeling a subject forming a random posterior belief, we assume that subject $i$ holds the first-order belief that others form random posterior beliefs with a Bayesian kernel.

\begin{ass}
	\label{ass:distribution-other-p}
	Subject $i$ believes that other subjects' posterior belief $\tilde{p}_x^{-i}$ follows a beta distribution with two parameters $\theta^i \ge 0$ and $p_x \in [0,1]$. Specifically,
	\begin{equation*}
	\tilde{p}_x^{-i} \sim \operatorname{Beta}\left(\frac{p_x}{\theta^i}, \frac{1-p_x}{\theta^i}\right) 
	\end{equation*}
	when $\theta^i >0$ and $p_x \in (0,1)$. $\tilde{p}_x^{-i} =p_x$ almost surely when $\theta^i =0$, $\tilde{p}_x^{-i} =0$ almost surely when $ p_x =0$, and $\tilde{p}_x^{-i} =1$ almost surely when $p_x =1$.
\end{ass}

Similar to Assumption \ref{ass:distribution-p}, Assumption \ref{ass:distribution-other-p} suggests that subject $i$ thinks that other subjects' posterior beliefs are unbiased on average. The parameter $\theta^i$ describes subject $i$'s opinion of the degree of others' belief errors. The larger $\theta^i$ is, the greater she thinks others' belief errors are. In the special case of $\theta^i= 0$, $\tilde{p}_x^{-i} = p_x$ almost surely because its distribution is degenerate. Therefore, subject $i$ believes that others' posterior beliefs all coincide with the Bayesian posterior belief, so their first choices are perfectly aligned with the colors of the balls they observe. In this case, the observation of a subject choosing urn 1/urn 2 is the same as observing a black/white ball, and the social information setting degenerates to the private information setting.

Moreover, we know that under Assumption \ref{ass:distribution-other-p}, $P^{i}(C_2\mid B) = P^{i}(C_1\mid W)>0$ when $\theta^i >0$, which implies that $P^{i}(C_1\mid\urnone) < P(B\mid\urnone)$ and $P^{i}(C_2\mid\urntwo) < P(W\mid\urntwo)$. In other words, subject $i$ interprets that the signal quality is always lower when it is a signal of another subject's first choice than when it is a signal of the color of a drawn ball. Her observation of another subject's first choice is equivalent to observing a signal with discounted quality.

Finally, the belief error model in the social information setting employs Assumption \ref{ass:distribution-other-p} to determine a subject's interpretation of others' first choices. The model uses both Assumption \ref{ass:distribution-p} and Assumption \ref{ass:distribution-other-p} to determine a subject's posterior belief after observing others' first choices. These two assumptions and the assumption of subjects' sophisticated consideration of belief error are needed when modeling subjects' extrapolation in the later stages.\footnote{To obtain a better picture of the deviation of our belief error model from the Bayesian benchmark model in the social information setting, one may introduce two intermediate non-Bayesian paradigms. Specifically, non-Bayesian paradigm 1 assumes that subjects form posterior beliefs in a Bayesian manner but believe that others form random posterior beliefs with a Bayesian kernel. Non-Bayesian paradigm 2 assumes that subjects form random posterior beliefs with a Bayesian kernel and believe that others use a similar method. Non-Bayesian paradigm 3, which is our belief error model in the social information setting, assumes that subjects make a sophisticated consideration of their following non-Bayesian paradigm 2. The introduction of Assumption \ref{ass:distribution-other-p} makes the model deviate from the Bayesian benchmark to non-Bayesian paradigm 1. Further introduction of Assumption \ref{ass:distribution-p} makes the model deviate from non-Bayesian paradigm 1 to non-Bayesian paradigm 2. The final introduction of the assumption of sophisticated consideration makes the model deviate from non-Bayesian paradigm 2 to our belief error model. }

Let $V_{B+1}^{i,soc}$ be subject $i$'s ex ante expected payoff of observing another subject's first choice after having already observed a black ball in the first stage. We simplify the notation to $V_{B+1}^{soc}$ whenever it is clear.\footnote{$
V_{B+1}^{soc} = \tilde{p}_B^i [P^{i}(C_1\mid\text{urn 1}) \cdot \mathbb{E}\max(\tilde{p}_{BC_1}^i, 1-\tilde{p}_{BC_1}^i)
+ P^{i}(C_2\mid\text{urn 1}) \cdot \mathbb{E}\max(\tilde{p}_{BC_2}^i, 1-\tilde{p}_{BC_2}^i) ] + (1-\tilde{p}_B^i) [P^{i}(C_2\mid\text{urn 2}) \cdot \mathbb{E}\max(\tilde{p}_{BC_2}^i, 1-\tilde{p}_{BC_2}^i) +  P^{i}(C_1\mid\text{urn 2}) \cdot \mathbb{E}\max(\tilde{p}_{BC_1}^i, 1-\tilde{p}_{BC_1}^i)].$
} Similar to the private information setting, we show that there is positive informational value of observing another subject's first choice, which contradicts the prediction according to the Bayesian benchmark model.
\begin{proposition}
Assume that sophisticated subjects with belief error maximize expected payoff. Then under Assumptions \ref{ass:distribution-p} and \ref{ass:distribution-other-p},
\[ V_{B+1}^{soc} - \max(\tilde{p}_B^i, 1-\tilde{p}_B^i) \begin{cases}
> 0 & \text{if $\tilde{p}_B^i \in (\ubar{p}', \bar{p}')$}, \\
\leq 0 & \text{if $\tilde{p}_B^i \leq \ubar{p}'$ or $\tilde{p}_B^i \geq \bar{p}'$},
\end{cases} \]
where $\ubar{p}'$, $\bar{p}'$ satisfy that $\ubar{p}' < \frac{1}{2} < p_B < \bar{p}'$.
\label{prop:social-information-value}
\end{proposition}
\begin{proof}
See Appendix \ref{pf:positive-social-info-value}.
\end{proof}
Similar to the private information setting, the difference in expected payoffs is generally positive from the perspective of an outside observer. Let $\mathbb{E}V_{B+1}^{soc}$ be a subject's expected payoff of observing another subject's first choice in the first stage from the perspective of an observer. Then, $\mathbb{E}V_{B+1}^{soc} = [p\pi +q(1-\pi) ]v(p_{BC_1}) + [p(1-\pi)+q\pi]v(p_{BC_2})$, where $\pi \equiv P^{i}(C_1\mid\text{urn 1})=P^{i}(C_2\mid\text{urn 2})$. Since $[p\pi +q(1-\pi)]p_{BC_1} + [p(1-\pi)+q\pi]p_{BC_2} =p_B$, it is clear that $\mathbb{E}V_{B+1}^{soc} > v(p_B) $ if $v(p_x)$ is convex.

We finally discuss the implication of the belief error model for different valuations of social information and private information. Consider a random variable $X$ that takes the value $p_{BB}$ with probability $p^2 + q^2$ and the value $p_{BW}$ with probability $2pq$, and another random variable $Y$ that takes the value $p_{BC_1}$ with probability $p\pi +q(1-\pi)$ and takes the value $p_{BC_2}$ with probability $p(1-\pi)+q\pi$. It is straightforward to check that $X$ is a mean-preserving spread of $Y$ since $\pi < p$ under Assumption \ref{ass:distribution-other-p} and $\mathbb{E} X = \mathbb{E} Y =p$. Therefore, a sufficient condition for $\mathbb{E}V_{B+1}^{pri} > \mathbb{E}V_{B+1}^{soc}$ is that $v(p_x)$ is convex in $p_x$. On the basis of the same logic, it is also a sufficient condition for $\mathbb{E}V_{B+3}^{pri} > \mathbb{E}V_{B+3}^{soc}$ in the case of three additional observations. Since our simulation exercise confirms that the convexity of $v(p_x)$ generally holds, the belief error model predicts that social information is less valuable than private information.

\subsubsection{Paying to seek confidence}
\label{sec:bidding-strategy}

We have shown that when subjects form random posterior beliefs with a Bayesian kernel, there is positive informational value for observing one additional signal. On the one hand, the positive informational value of additional signals due to belief error has a clear upper bound, e.g., $V_{B+1}^{pri} - \max(\tilde{p}_B^i, 1-\tilde{p}_B^i) < \frac{1}{2}$ when subjects are assumed to maximize expected payoff. This upper bound corresponds to a maximum  reservation price of 150 tokens in our experiment. On the other hand, subjects' reported reservation price in each treatment condition of our experiment ranges from 0 tokens to 300 tokens. Therefore, an additional assumption is required to rationalize the full data sample.

We adopt the assumption that subjects paying to seek confidence, which has been demonstrated in \citet{eliaz:10}.\footnote{One may naturally wonder if introducing the assumption of risk aversion could rationalize subjects' considerably high reservation price for additional signals in the data. While risk aversion does help to generate a higher reservation price than under the assumption of risk neutrality, the additional assumption is still not sufficient because we then need an unreasonably high coefficient of risk aversion to rationalize the reported reservation prices that are much higher than 150 tokens. Moreover, no coefficient of risk aversion is able to rationalize the reported reservation price of 300 tokens, which is the monetary reward of making a correct choice.} In their experimental study, they provide compelling evidence that subjects are willing to pay for information on the likelihood that a decision is ex post optimal. They propose an explanation that subjects have an intrinsic preference for being ``confident'' in choosing the right decision. In a similar vein, we assume that a subject's utility consists of two components: the monetary reward from making a correct choice, and the psychological reward of making a correct choice. For simplicity, we assume that the second component of the utility is proportional to the chance of making a correct choice, and the interpretation is that the higher the chance is, the better the anticipatory feelings subjects have.

\begin{ass}
	\label{ass:information-hoarding}
	In addition to gaining utility from the monetary reward, subject $i$ gains utility from the confidence in earning the monetary reward, which is assumed to be proportional to the chance of earning the monetary reward, characterized by a parameter $\alpha ^i \geq 0$.
\end{ass}

Suppose that subject $i$ has an endowment $w_0 > 0$ and pays $s$ to observe additional signals. The reward of making a correct choice is $r > 0$, and the reward of an incorrect choice is zero. Then, according to Assumption \ref{ass:information-hoarding}, her expected utility from choosing urn 1 given signal set $x$ and posterior belief $\tilde{p}_x^i$ is $\tilde{p}_x^i (w_0 +r-s + \alpha ^i) + (1-\tilde{p}_x^i )(w_0 -s) = w_0 -s +\tilde{p}_x^i (r+\alpha ^i)$. Correspondingly, her expected utility from choosing urn 2 is $w_0 -s +(1-\tilde{p}_x^i )(r+\alpha ^i)$.

We now investigate subject $i$'s optimal strategy of bidding, i.e., the optimal reservation price for additional signals after observing the first signal. For simplicity, let us consider the scenario in which the first ball observed is black ($B$) since it can be shown that the subject employs the same bidding strategy in the scenario of first ball being white ($W$). Let $\mathcal{X}_{B+n}^t$ denote the collection of all signal sets $x$ that contains a first black ball in setting $t \in \{pri,soc\}$ with $n \in \{1, 3\}$ additional signals. For example, when the additional signals consist of the color of one ball, $\mathcal{X}_{B+1}^{pri} = \{\BB, \BW\}$; when the additional signals consist of the colors of three balls, $\mathcal{X}_{B+3}^{pri} = \{4B, 3B1W, 2B2W, 1B3W\}$.

By experimental design, subject $i$ bids $b$ and eventually pays the ask price $s$ to observe additional signals when successfully purchasing information, i.e., when $s \leq b$. When the ask price $s$ is uniformly distributed on $[0, w_0]$, sophisticated subject $i$ with a belief of urn 1,  $\tilde{p}_B^i$, has an expected utility from bidding $b$ as follows,
\begin{align*}
    U_{n,t}(b,\tilde{p}_B^i) & = \int_0^b\frac{1}{w_0}\sum\nolimits_{x\in\mathcal{X}_{B+n}^t} P^i(x\mid B)\Bigl\{(w_0-s+r+\alpha ^i) \mathbb{E}\max(\tilde{p}_x^i, 1-\tilde{p}_x^i) \\
    & \qquad + (w_0-s) \left[1-\mathbb{E}\max(\tilde{p}_x^i, 1-\tilde{p}_x^i)\right]\Bigr\}\od s \\
    & \qquad + \int_b^{w_0}\frac{1}{w_0}\Bigl\{(w_0+r+\alpha ^i) \max(\tilde{p}_B^i, 1-\tilde{p}_B^i)+w_0 \left[1-\max(\tilde{p}_B^i, 1-\tilde{p}_B^i)\right]\Bigr\}\od s \\
    &= \frac{1}{w_0}\int_0^b\biggl\{(w_0-s+r+\alpha ^i) \sum\nolimits_{x\in\mathcal{X}_{B+n}^t} P^i(x\mid B)\mathbb{E}\max(\tilde{p}_x^i, 1-\tilde{p}_x^i) \\
    & \qquad + (w_0-s) \biggl[1-\sum\nolimits_{x\in\mathcal{X}_{B+n}^t}P^i(x\mid B)\mathbb{E}\max(\tilde{p}_x^i, 1-\tilde{p}_x^i)\biggr]\biggr\}\od s \\
    & \qquad + \frac{w_0-b}{w_0}\Bigl\{(w_0+r+\alpha ^i) \max(\tilde{p}_B^i, 1-\tilde{p}_B^i) + w_0 \left[1-\max(\tilde{p}_B^i, 1-\tilde{p}_B^i)\right]\Bigr\},
\end{align*}
where $P^i(x\mid B)$ refers to subject $i$'s belief about observing signal set $x$ conditional on observing the first black ball, and $\sum_{x\in\mathcal{X}_{B+n}^t} P^i(x\mid B) =1$.

Define $V_{B+n}^{i,t} = \sum_{x\in\mathcal{X}_{B+n}^t} P^i(x\mid B)\mathbb{E}\max(\tilde{p}_x^i, 1-\tilde{p}_x^i)$ and we simplify the notation to $V_{B+n}^{t}$ whenever it is clear. Since $\mathbb{E}\max(\tilde{p}_x^i, 1-\tilde{p}_x^i)$ refers to subject $i$'s probability of making a correct choice after observing signal set $x$ before the posterior belief is realized,
we interpret $V_{B+n}^t $ as the ``average'' or ``expected'' probability of making a correct choice by purchasing additional signals. We show that $V_{B+n}^t$ is linear in the first-stage posterior belief $\tilde{p}_B^i$ and is increasing in $\tilde{p}_B^i$ in Appendix \ref{pf:expected-win}.

We then characterize subject $i$'s optimal bidding strategy below.

\begin{proposition}
Under Assumptions \ref{ass:distribution-p}, \ref{ass:distribution-other-p},  and \ref{ass:information-hoarding}, (i) when $\alpha ^i \leq \frac{2w_0}{V_1+V_2-1}-r$, sophisticated subject $i$'s optimal bidding strategy is\\
\begin{equation*}
    b(\tilde{p}_B^i) = \begin{cases}
    0 & \text{if $\tilde{p}_B^i < \frac{1-V_2}{1+V_1-V_2}$ or $\tilde{p}_B^i > \frac{V_2}{1+V_2-V_1}$}, \\
    (r+\alpha ^i)[V_{B+n}^t-\max(\tilde{p}_B^i, 1-\tilde{p}_B^i)] & \text{if $\tilde{p}_B^i \in \left[\frac{1-V_2}{1+V_1-V_2},\frac{V_2}{1+V_2-V_1}\right]$},
    \end{cases}
\end{equation*}
and (ii) when $\alpha ^i > \frac{2w_0}{V_1+V_2-1}-r$, her optimal bidding strategy is
\[
    b(\tilde{p}_B^i) = \begin{cases}
    0 & \text{if $\tilde{p}_B^i < \frac{1-V_2}{1+V_1-V_2}$ or $\tilde{p}_B^i > \frac{V_2}{1+V_2-V_1}$}, \\
    (r+\alpha^i)[V_{B+n}^t-\max(\tilde{p}_B^i, 1-\tilde{p}_B^i)] & \text{if $\tilde{p}_B^i \in \left[\frac{1-V_2}{1+V_1-V_2},\frac{1-V_2+\Delta}{1+V_1-V_2}\right]\cup\left[\frac{V_2-\Delta}{1+V_2-V_1},\frac{V_2}{1+V_2-V_1}\right]$}, \\
    w_0 & \text{if $\tilde{p}_B^i \in \left(\frac{1-V_2+\Delta}{1+V_1-V_2},\frac{V_2-\Delta}{1+V_2-V_1}\right)$},
    \end{cases}
\]
where $\Delta = w_0/(r+\alpha^i)$ and $V_k = \sum_{x\in\mathcal{X}_{B+n}^t} P^i(x\mid \mbox{urn } k, B)\cdot\mathbb{E}\max(\tilde{p}_x^i, 1-\tilde{p}_x^i)$ for $k\in\{1,2\}$. In addition, the bidding function $b(\tilde{p}_B^i)$ is non-decreasing on $\bigl[0,\frac{1}{2}\bigr]$ and non-increasing on $\bigl[\frac{1}{2},1\bigr]$.
\label{prop:bidding-strategy}
\end{proposition}
\begin{proof}
See Appendix \ref{pf:bidding-strategy}.
\end{proof}

Proposition \ref{prop:bidding-strategy} shows that subject $i$ bids the highest when her posterior belief in the first stage is close to $\frac{1}{2}$ and bids the lowest when the posterior belief is close to $0$ or $1$. The interpretation is that the additional signals are the most helpful when she feels uncertain about the underlying state, i.e., $\tilde{p}_B^i  = \frac{1}{2}$. Moreover, the additional signals are the least helpful when she feels certain about the underlying state, i.e., $\tilde{p}_B^i  = 1$ or $0$. Furthermore, in the case of $  b(\tilde{p}_B^i)  = (r+\alpha ^i)[V_{B+n}^t - \max(\tilde{p}_B^i, 1-\tilde{p}_B^i)]$, subject $i$'s bid may be naturally decomposed into two parts: $r [V_{B+n}^t - \max(\tilde{p}_B^i, 1-\tilde{p}_B^i)]$ and $\alpha ^i [V_{B+n}^t - \max(\tilde{p}_B^i, 1-\tilde{p}_B^i)]$. The first part is the expected increment in monetary reward due to the increment in the chance of making a correct choice, which we call the instrumental value of additional information. The second part is the expected increment in psychological reward due to the increment in the chance of making a correct choice, which we label the non-instrumental value of additional information.

\subsection{Heterogeneous belief error model}
Individual-level analysis of the experimental data in Section \ref{sec:individual-level-analysis} demonstrates that considerable heterogeneity exists among subjects, especially in terms of the valuations of private information and social information. Since a basic model of belief error is captured by three parameters, our heterogeneous belief error model assumes that subjects are heterogeneous in that the model parameters $(\gamma^i, \theta^i, \alpha^i)$ vary across subjects. Specifically, we assume that the model parameters for each subject are independently drawn from a certain distribution and that the realization of the values is privately observable to the subject. Since a natural interpretation of the model parameters requires that $\gamma^i \geq 0$, $\theta^i \geq 0 , \alpha ^i \geq 0$, we make the following assumption. 

\begin{ass}
\label{ass:heterogeneity}
$(\gamma^i, \theta^i, \alpha ^i)$, $i = 1, 2, \ldots, N$, are independently and identically distributed. $\gamma^i$, $\theta^i$, and $\alpha ^i$ are jointly independent and follow exponential distributions with means $\bar{\gamma}$, $\bar{\theta}$, and $\bar{\alpha}$, respectively. Specifically, the probability density function of $(\gamma^i,\theta^i,\alpha ^i)$ is
\[
    \phi(\gamma^i, \theta^i, \alpha ^i) = \frac{1}{\bar{\gamma} \bar{\theta} \bar{\alpha}}\exp\left(-\frac{\gamma^i}{\bar{\gamma}}-\frac{\theta^i}{\bar{\theta}}-\frac{\alpha^i}{\bar{\alpha}}\right), \quad \gamma^i, \theta^i, \alpha^i \geq 0,
\]
when $\bar{\gamma} >0$, $\bar{\theta} > 0$, and $\bar{\alpha}> 0$; if any of $\bar{\gamma}$, $\bar{\theta}$, or $\bar{\alpha}$ is zero, the distribution of corresponding model parameter is degenerate with all probability mass at 0.
\end{ass}

According to the heterogeneous belief error model, the model parameters $(\gamma^i, \theta^i, \alpha^i)$ remain constant for a subject once they are drawn from the exponential distribution. Given the parameters $(\gamma^i, \theta^i)$, the subject forms her posterior belief $\tilde{p}_x^i$ and her first-order belief about others' posterior belief $\tilde{p}_x^{-i}$ for any possible signal set $x$ according to the corresponding beta distributions specified in Assumptions \ref{ass:distribution-p} and \ref{ass:distribution-other-p}. We assume that the subject's belief error and her belief about others' belief error are independent of each other and are both independent of the realization of the parameter $\alpha^i$ that characterizes the psychological reward of making a correct choice.

\begin{ass}
For any $i \in \{1, 2, \ldots, N\}$, any realization of $(\gamma^i, \theta^i)$, and any possible signal set $x$, $\tilde{p}_x^i$ , $\tilde{p}_x^{-i}$, and $\alpha^i$ are independent of each other.
\end{ass}

\subsection{Estimation strategy and results}
\label{estimationresult}
We apply the heterogeneous belief error model to estimate the model parameters $(\bar\gamma, \bar\theta, \bar{\alpha})$.\footnote{Our choice of structural estimation of the full data sample based on the heterogeneous model instead of the basic model is made due mainly to the following reasons. First, subjects in feedback sessions and no-feedback sessions may have different model parameters due to the possibility of learning from the feedback. Second, we find substantial heterogeneity across subjects, as illustrated in Section \ref{sec:individual-level-analysis}. Third, we apply a basic model of belief error to estimate the model parameters for each subject's data subsample, and the estimation results show large variation in the estimated values across subjects.} The data consist of 6000 observations from $N = 100$ subjects in total $J = 60$ rounds. Let $c^i_j$ and $b^i_j$ denote subject $i$'s first choice and her reservation price for additional signals in the $j$-th round. Let $\mathcal{J}_i\subseteq\{1,2,\ldots,J\}$ denote the collection of rounds in which subject $i$ successfully purchases additional signals, and let $d^i_j$ ($j \in \mathcal{J}_i$) be her second choice after the additional signals are observed. Then, given $(\gamma^i,\theta^i,\alpha^i)$, the probability (or likelihood) that subject $i$ chooses $c^i_j$ ($j \in \{1,2,\ldots,J\}$) and $d^i_j$ ($j\in\mathcal{J}_i$) and reports a reservation price of $b^i_j$ ($j \in \{1,2,\ldots,J\}$) is
\[
    L_i^*(\gamma^i,\theta^i, \alpha^i) = \prod_{j = 1}^J f_j(c^i_j,b^i_j\mid \gamma^i, \theta^i, \alpha ^i) \cdot \prod_{j\in\mathcal{J}_i} g_j(d^i_j\mid \gamma^i,\theta^i),
\]
where $f_j$ represents the probability of the first choice and the reservation price in round $j$, and $g_j$ represents the probability of the second choice. Integrating out the unobservable individual-specific $\gamma^i$, $\theta^i$ and $\alpha ^i$ yields the likelihood for subject $i$'s behavioral data given model parameters $(\bar{\gamma},\bar{\theta},\bar{\alpha})$ as
\begin{multline}
    L_i(\bar{\gamma},\bar{\theta},\bar{\alpha}) = \iiint_{\mathbb{R}^3_+} L_i^*(\gamma^i,\theta^i,\alpha^i) \phi(\gamma^i, \theta^i, \alpha ^i) \od\gamma^i \od\theta^i \od \alpha ^i \\
    = \iiint_{\mathbb{R}^3_+} \left[ \, \prod_{j = 1}^J f_j(c^i_j,b^i_j\mid \gamma^i, \theta^i, \alpha ^i) \prod_{j\in\mathcal{J}_i} g_j(d^i_j\mid \gamma^i,\theta^i) \right] \phi(\gamma^i, \theta^i, \alpha ^i) \od\gamma^i \od\theta^i \od \alpha ^i.
    \label{eq:individual-llk}
\end{multline}
We use maximum likelihood estimation to obtain the estimates of $(\bar{\gamma},\bar{\theta},\bar{\alpha})$ and derive the likelihood function in detail in the supplementary Appendix B.

\[ (\hat{\bar\gamma}, \hat{\bar\theta}, \hat{\bar{\alpha}}) = \argmax_{(\bar\gamma, \bar\theta, \bar{\alpha})\in\mathbb{R}^3_+} \sum_{i=1}^N \ln L_i(\bar\gamma,\bar\theta,\bar \alpha). \]

To compute the triple integral in the likelihood function, we use the three-dimensional Gauss-Jacobi quadrature with 50 nodes in each dimension. In the estimation, we also normalize the subjects' reservation prices to the unit interval by dividing the reservation prices in units of token by the upper bound of the reservation price (i.e., 300 tokens). The estimates of the model parameters are reported in column 1 of Table \ref{tab:estimation}, with the standard errors reported in parentheses.\footnote{As discussed in the beginning of Section \ref{sec:belief-disturbance-model}, we also apply a heterogeneous logit QRE model to perform structural estimation by pretending to ignore the conceptual concern. Specifically, the heterogeneous logit QRE model assumes that subject $i$'s payoff disturbance follows a Gumbel distribution with a scale parameter $\beta ^{i}$, that in her opinion, others' payoff disturbances follow a Gumbel distribution with a potentially different scale parameter $\beta _{0}^{i}$, and that the two parameters are independently drawn from exponential distributions with means $\bar\beta$ and $\bar\beta _0$. The maximum likelihood estimation shows that the estimate of $\bar\beta _0$ is significantly larger than the estimate of $\bar\beta $. In addition, the log-likelihood value of the estimation is -2879.85, which is considerably smaller than the log-likelihood value of the estimation based on our heterogeneous belief error model.}
\begin{table}[htbp]
\centering
\caption{Estimation results of model parameters}
\label{tab:estimation}
\begin{tabular}{cccc}
\toprule
 & \begin{tabular}[c]{@{}c@{}}(1)\\ Unrestricted\end{tabular} & \begin{tabular}[c]{@{}c@{}}(2)\\ Test $H_0: \bar\theta = 0$\end{tabular} & \begin{tabular}[c]{@{}c@{}}(3)\\ Test $H_0: \bar\theta \leq \bar\gamma$\end{tabular} \\ \midrule
$\bar\theta$ & 0.2173 & - & 0.1677 \\[-0.5ex]
 &(0.0367) &  & (0.0321) \\[0.5ex]
$\bar\gamma$ & 0.1624 & 0.1551 & 0.1677 \\[-0.5ex]
 & (0.0014) & (0.0078) & (0.0084) \\[0.5ex]
$\bar \alpha $ & 2.1692 & 2.1385 & 2.1537 \\[-0.5ex]
 & (0.1382) & (0.1359) & (0.1381) \\ \midrule
Log-likelihood & $-2163.36$ & $-2176.45$ & $-2164.77$ \\[0.5ex]
$p$-value & & $1.552\times 10^{-7}$ & 0.0467 \\
\bottomrule
\end{tabular}
\end{table}

Recall that parameter $\bar\gamma$ measures the degree of an average subject's belief error and that parameter $\bar\theta$ measures an average subject's belief about the degree of others' belief errors. A larger estimate of $\bar\theta$ relative to $\bar\gamma$ suggests that an average subject thinks that other people have more noise in forming their posterior beliefs than actually exists.\footnote{\citet{kubler:04} estimate a logistic choice model and find that the estimated value of a subject's belief about others' choice disturbance is greater than that of her own choice disturbance. \citet{goeree:07} estimate a quantal response model with non-rational expectations and find that a subject's belief about others' payoff disturbance is greater than that of her own payoff disturbance. In contrast, the estimation of our model suggests that subjects view others' belief errors as greater than theirs or greater than what it actually is. Our interpretation is conceptually different from theirs.} Thus, a subject's lower valuation of observing social information is explained both by her awareness of others' belief errors and by her exaggeration of others' belief errors. In other words, subjects' taking into account the belief error leads to their discounting the quality of the signals of others' choices, and their exaggeration of others' belief errors leads to a further discount in the signal quality.

Finally, we conduct statistical tests to confirm the two forces driving the difference in reservation prices of private information and social information. When $\bar\theta$ approaches zero, $\theta^i$ becomes a degenerate distribution, and $\theta^i = 0$ for any $i$. According to the belief error model, subject $i$ does not think others have belief error in this case. Therefore, she believes that others' first choices perfectly coincide with the colors of the balls observed privately; thus, social information is equivalent to private information. This possibility is ruled out by testing
\[
	H_0: \bar\theta = 0 \quad \text{versus} \quad H_1: \bar\theta > 0.
\]
We use the generalized likelihood ratio test. Under the null hypothesis, the likelihood ratio test statistic
\[
	2\left[\max\nolimits_{\bar\gamma,\bar\theta,\bar \alpha \geq 0} \sum_{i=1}^N \ln L_i(\bar\gamma,\,\bar\theta,\,\bar \alpha) - \max\nolimits_{\bar\theta = 0,\ \bar\gamma,\,\bar \alpha \geq 0} \sum_{i=1}^N \ln L_i(\bar\gamma,\bar\theta,\bar \alpha)\right] \xrightarrow{d} \frac{1}{2}\chi^2_0 + \frac{1}{2}\chi^2_1,
\]
where $\chi^2_0$ is a degenerate distribution at 0 and $\chi^2_1$ is a chi-square distribution with degree of freedom 1. Given that the value of the test statistic is 26.1839 ($p$-value $=1.552 \times 10^{-7}$), the null hypothesis is rejected at the significance level of 1\%.

We also conduct the test
\[ H_0: \bar\theta \leq \bar\gamma \quad \text{versus} \quad H_1: \bar\theta > \bar\gamma \]
to verify the interpretation that subjects exaggerate others' belief error. Similarly, the likelihood ratio test statistic
\[
	2\left[\max\nolimits_{\bar\gamma,\bar\theta,\bar\alpha \geq 0} \sum_{i=1}^N \ln L_i(\bar\gamma,\,\bar\theta,\,\bar\alpha) - \max\nolimits_{\bar\gamma \geq \bar\theta \geq 0,\ \bar\alpha \geq 0} \sum_{i=1}^N \ln L_i(\bar\gamma,\bar\theta,\bar\alpha)\right] \xrightarrow{d} \frac{1}{2}\chi^2_0 + \frac{1}{2}\chi^2_1
\]
under the null hypothesis. Since the $p$-value of this test is 0.0467, the test rejects the null hypothesis at the significance level of 5\%.

\section{Conclusions}
\label{sec:conclusion}
This paper reveals experimentally that individuals value social information less than private information, even though they are expected to be identical in the Bayesian paradigm. Additionally, a monotonic relationship exists between signal quality and the frequency of individuals' making a Bayesian choice, and there is positive informational value of observing an additional signal after already observing a signal, both of which contradict the Bayesian paradigm. These findings are explained by a belief error model in which individuals form a random posterior belief with a Bayesian kernel and individuals sophisticatedly consider their and others' belief errors . Finally, maximum likelihood estimation of the heterogeneous belief error model suggests that individuals' sophisticated consideration of others' belief errors and their exaggeration of others' belief errors both contribute to their lower valuation of social information than private information.

This paper, to the best of our knowledge, is the first to test the first-order belief assumption that individuals believe that others process private information in a Bayesian manner. In our novel experimental design, the first-order belief assumption is identified with a testable implication about the equivalent reservation prices for private information and social information. Our experimental evidence casts doubt on the first-order belief assumption and suggests that future non-Bayesian social learning models may need to reflect the failure of the assumption given that it is fundamental to many existing Bayesian and non-Bayesian social learning models.

Our proposed belief error model first formalizes the noise in individuals' formation of posterior beliefs by making a beta distribution assumption about their random posterior beliefs. Compared to the method of modeling errors by assuming an additively separable error term and making certain distribution assumption about the error term, we believe that our method is particularly useful for modeling belief error and, more generally, for modeling errors in economic variables whose values must fall within a bounded interval. In addition, our model has the advantage of retaining the feature that individuals are, on average, Bayesian and is flexible for allowing some non-Bayesian choice behavior that cannot be predicted by existing non-Bayesian models. It is beyond the scope of this paper to provide a rational foundation for the beta distribution assumption about random posterior beliefs. We believe that an investigation of its rational foundation is an interesting research agenda, just as \citet{matejka:15} recently provides a rational inattention foundation for the logit model that has been used for decades.

\clearpage

\appendix

\section{Proofs of Main Results}

\subsection{Proof of Proposition \ref{prop:subjective-assessment}}
\label{pf:subjective-assessment}

\begin{proof}
	For simplicity of notation, we omit the index $i$ in the proof. Given Assumption \ref{ass:distribution-p}, $\tilde{p}_x$ and $\tilde{p}_y$ follow the same distribution when $p_x = p_y$, so $\max(\tilde{p}_x, 1-\tilde{p}_x)$ and $\max(\tilde{p}_y, 1-\tilde{p}_y)$ are also identically distributed. Then their expectation over belief error, i.e. $P(\text{correct choice}\mid x)$ and $P(\text{correct choice}\mid y)$, should be equal.

	According to the property of beta distribution,
	\[ \tilde{p}_y \sim \operatorname{Beta}\left(\frac{p_y}{\gamma}, \frac{1-p_y}{\gamma}\right) \ \Rightarrow\  1-\tilde{p}_y \sim \operatorname{Beta}\left(\frac{1-p_y}{\gamma},\frac{p_y}{\gamma}\right). \]
	If $p_y = 1-p_x$, then $1-\tilde{p}_y$ and $\tilde{p}_x$ are identically distributed.  So $\max(\tilde{p}_y, 1-\tilde{p}_y) = \max(1-\tilde{p}_y, \tilde{p}_y)$ and $\max(\tilde{p}_x, 1-\tilde{p}_x)$ are identically distributed and in turn their expectation over belief error should be equal. This establishes (i).

We now prove the monotonicity. According to property (i), it suffices to show that $p_x > p_y \geq \frac{1}{2}$ implies that $P(\text{correct choice}\mid x) >P(\text{correct choice}\mid y)$.

Since $P(\text{correct choice}\mid x) = \mathbb{E}\max(\tilde{p}_x, 1-\tilde{p}_x)$ where $\tilde{p}_x \sim \operatorname{Beta}\left(\frac{p_x}{\gamma}, \frac{1-p_x}{\gamma}\right)$, we define random variables $X=\max(\tilde{p}_x, 1-\tilde{p}_x)$ and $Y=\max(\tilde{p}_y, 1-\tilde{p}_y)$ which have common support $\left[\frac{1}{2},1\right]$. Then, the proposition states that $\mathbb{E}X > \mathbb{E}Y$ whenever $p_x > p_y \geq \frac{1}{2}$. This is implied by that $X$ first order stochastically dominates $Y$; that is,
\begin{equation}
	P(X\leq u) < P(Y\leq u) \quad \text{for any } u \in \left(\frac{1}{2}, 1\right).
	\label{eq:fosd}
\end{equation}

We shall show \eqref{eq:fosd} by showing that for random variable $T \sim \operatorname{Beta}\left(\frac{p}{\gamma}, \frac{1-p}{\gamma}\right)$ with $p > \frac{1}{2}$ and for any $\frac{1}{2} < u < 1$, $P(\max(T, 1-T) \leq u)$ is strictly decreasing in $p$, or equivalently,
\begin{equation}
	\frac{\partial P(\max(T, 1-T)\leq u)}{\partial p} < 0 \quad \text{for any $p > \frac{1}{2}$ and any $\frac{1}{2} < u < 1$}.
	\label{eq:fosd-equivalent}
\end{equation}

First, for the beta function $\Beta(a,b) = \Gamma(a)\Gamma(b)/\Gamma(a+b)$ with $a,b > 0$,
\[
	\frac{\partial\Beta(a,b)}{\partial a} = \Beta(a,b)[\psi(a)-\psi(a+b)], \quad \frac{\partial\Beta(a,b)}{\partial b} = \Beta(a,b)[\psi(b)-\psi(a+b)]
\]
where $\psi(z) = \Gamma'(z)/\Gamma(z)$ is the digamma function. Then for $0<z<1$,
\begin{equation*}
	\frac{\partial\ibeta_z(a,b)}{\partial a} = \frac{\partial}{\partial a}\frac{\int_0^z t^{a-1}(1-t)^{b-1}\od t}{\Beta(a,b)} = \frac{\int_0^z t^{a-1}(1-t)^{b-1}\ln t\od t}{\Beta(a,b)} - \ibeta_z(a,b)[\psi(a)-\psi(a+b)],
\end{equation*}
and
\begin{equation*}
	\frac{\partial\ibeta_z(a,b)}{\partial b} = \frac{\int_0^z t^{a-1}(1-t)^{b-1}\ln(1-t)\od t}{\Beta(a,b)} - \ibeta_z(a,b)[\psi(b)-\psi(a+b)],
\end{equation*}
which gives
\begin{equation}
	\frac{\partial\ibeta_z\left(\frac{p}{\gamma}, \frac{1-p}{\gamma}\right)}{\partial p} = \frac{1}{\gamma}\left\{\int_0^z f_T(t)\ln\frac{t}{1-t}\od t - \ibeta_z\left(\frac{p}{\gamma}, \frac{1-p}{\gamma}\right)\left[\psi\left(\frac{p}{\gamma}\right) - \psi\left(\frac{1-p}{\gamma}\right)\right]\right\},
	\label{eq:di/dp}
\end{equation}
where $f_T(t) = t^{\frac{p}{\gamma}-1} (1-t)^{\frac{1-p}{\gamma}-1} \big/ \Beta\left(\frac{p}{\gamma}, \frac{1-p}{\gamma}\right)$ is the density function of $T$. Therefore, for $\frac{1}{2} < u < 1$, $P(\max(T,1-T)\leq u) = P(1-u \leq T \leq u) = \ibeta_u\left(\frac{p}{\gamma}, \frac{1-p}{\gamma}\right) - \ibeta_{1-u}\left(\frac{p}{\gamma}, \frac{1-p}{\gamma}\right)$, and then
\begin{multline}
	\frac{\partial P(\max(T,1-T)\leq u)}{\partial p}
= \frac{1}{\gamma}\left\{\int_{1-u}^u f_T(t)\ln\frac{t}{1-t}\od t - \int_{1-u}^u f_T(t)\od t\cdot\left[\psi\left(\frac{p}{\gamma}\right) - \psi\left(\frac{1-p}{\gamma}\right)\right]\right\} \\
	= \frac{\int_{1-u}^u f_T(t)\od t}{\gamma}
	\left\{\frac{\int_{1-u}^u f_T(t)\ln\frac{t}{1-t}\od t}{\int_{1-u}^u f_T(t)\od t} - \left[\psi\left(\frac{p}{\gamma}\right) - \psi\left(\frac{1-p}{\gamma}\right)\right]\right\}.
	\label{eq:dP/dp}
\end{multline}
Since $\int_{1-u}^u f_T(t)\od t > 0$ for any $u > \frac{1}{2}$, we have the sign of $\partial P(\max(T,1-T)\leq u)/\partial p$ is the same as the sign of the term
\begin{align*}
	A & \equiv \frac{\int_{1-u}^u f_T(t)\ln\frac{t}{1-t}\od t}{\int_{1-u}^u f_T(t)\od t} - \left[\psi\left(\frac{p}{\gamma}\right) - \psi\left(\frac{1-p}{\gamma}\right)\right] \notag \\
	&= \frac{\int_{1/2}^u [f_T(t)-f_T(1-t)]\ln\frac{t}{1-t}\od t}{\int_{1/2}^u [f_T(t)+f_T(1-t)]\od t} - \left[\psi\left(\frac{p}{\gamma}\right) - \psi\left(\frac{1-p}{\gamma}\right)\right].
\end{align*}

Next, differentiating $A$ with respect to $u$ yields
\begin{align}
	\frac{\partial A}{\partial u}
	&= \frac{f_T(u)+f_T(1-u)}{\int_{1/2}^u [f_T(t)+f_T(1-t)]\od t}\left\{ \frac{f_T(u)-f_T(1-u)}{f_T(u)+f_T(1-u)}\ln\frac{u}{1-u}
	\vphantom{\int_{1/2}^u \frac{f_T(t)-f_T(1-t)}{f_T(t)+f_T(1-t)}\ln\frac{t}{1-t}\cdot \frac{f_T(t)+f_T(1-t)}{\int_{1/2}^u [f_T(v)+f_T(1-v)]\od v}\od t}
	\right. \notag \\
	&\hspace{8em} - \left. \int_{1/2}^u \frac{f_T(t)-f_T(1-t)}{f_T(t)+f_T(1-t)}\ln\frac{t}{1-t}\cdot \frac{f_T(t)+f_T(1-t)}{\int_{1/2}^u [f_T(v)+f_T(1-v)]\od v}\od t \right\} \notag \\
	&= \frac{f_T(u)+f_T(1-u)}{\int_{1/2}^u [f_T(t)+f_T(1-t)]\od t}\left\{ g(u) - \int_{1/2}^u g(t)\cdot\frac{f_T(t)+f_T(1-t)}{\int_{1/2}^u [f_T(v)+f_T(1-v)]\od v}\od t \right\},
	\label{eq:dconditionalE-wrt-u}
\end{align}
where function $g(t) \equiv \frac{f_T(t)-f_T(1-t)}{f_T(t)+f_T(1-t)}\ln\frac{t}{1-t}$. Because $\frac{t}{1-t}$ is increasing in $t$ on $\left(\frac{1}{2},1\right)$,
\begin{equation*}
	\frac{f_T(t)-f_T(1-t)}{f_T(t)+f_T(1-t)} = \frac{t^{\frac{p}{\gamma}-1}(1-t)^{\frac{1-p}{\gamma}-1} - (1-t)^{\frac{p}{\gamma}-1}t^{\frac{1-p}{\gamma}-1}}{t^{\frac{p}{\gamma}-1}(1-t)^{\frac{1-p}{\gamma}-1} + (1-t)^{\frac{p}{\gamma}-1}t^{\frac{1-p}{\gamma}-1}}
	= 1-\frac{2}{1+\left(\frac{t}{1-t}\right)^{\frac{2p-1}{\gamma}}}
\end{equation*}
is increasing in $t$ given $p > \frac{1}{2}$; and so is $g(t)$ since $\ln\frac{t}{1-t}$ is increasing in $t$ on $\left(\frac{1}{2},1\right)$ as well. Note that the integral term in \eqref{eq:dconditionalE-wrt-u} is a weighted average of the values of $g(t)$ over interval $\left[\frac{1}{2},u\right]$. Thus, for any $u > \frac{1}{2}$,
\[
	g(u) > \int_{1/2}^u g(t)\cdot\frac{f_T(t)+f_T(1-t)}{\int_{1/2}^u [f_T(v)+f_T(1-v)]\od v}\od t
\]
due to the monotonicity of $g(\wc)$ and then $\partial A/\partial u < 0$.

Finally, when $u = 1$,\footnote{This result is also implied by \eqref{eq:dP/dp}: when $u = 1$, $P(\max(T,1-T)\leq 1) = 1$ for any $p$, so the left-hand side of \eqref{eq:dP/dp} is zero; it means the brace term on the right-hand side of \eqref{eq:dP/dp} must be zero since $\int_0^1 f_T(t)\od t / \gamma = 1/\gamma > 0$. For the expectation of $\ln \frac{T}{1-T}$, because for any $a,b > 0$,
\[
\int_0^1\frac{u^{a-1}(1-u)^{b-1}}{\Beta(a,b)}\ln u\od u = \frac{1}{\Beta(a,b)}\int_0^1\frac{\partial u^{a-1}(1-u)^{b-1}}{\partial a}\od u = \frac{1}{\Beta(a,b)}\frac{\partial\Beta(a,b)}{\partial a} = \psi(a)-\psi(a+b),
\]
it follows that $\mathbb{E}\ln T = \psi\left(\frac{p}{\gamma}\right) - \psi(1)$ and $\mathbb{E}\ln (1-T) = \psi\left(\frac{1-p}{\gamma}\right) - \psi(1)$.}
\begin{align*}
	A &= \frac{\int_0^1 f_T(t)\ln\frac{t}{1-t}\od t}{\int_0^1 f_T(t)\od t} - \left[\psi\left(\frac{p}{\gamma}\right) - \psi\left(\frac{1-p}{\gamma}\right)\right] \\
	&= \mathbb{E}\left(\ln\frac{T}{1-T}\right) - \left[\psi\left(\frac{p}{\gamma}\right) - \psi\left(\frac{1-p}{\gamma}\right)\right] = 0.
\end{align*}
Combining this and $\partial A/\partial u < 0$, we have $A < 0$
for any $\frac{1}{2} < u < 1$ and any $p > \frac{1}{2}$. This proves \eqref{eq:fosd-equivalent} and completes the proof.
\end{proof}

\subsection{Proof of Proposition \ref{prop:bayesian-monotone}}
\label{pf:bayesian-monotone}

\begin{proof}
	For simplicity of notation, we omit the index $i$ in the proof. We first show that $P\left(\tilde{p}_x \leq \frac{1}{2}\right)$ is strictly decreasing in $p_x$.
	
Given $\tilde{p}_x \sim \operatorname{Beta}\left(\frac{p_x}{\gamma}, \frac{1-p_x}{\gamma}\right)$, $P\left(\tilde{p}_x \leq \frac{1}{2}\right) = \ibeta_{\frac{1}{2}}\left(\frac{p_x}{\gamma}, \frac{1-p_x}{\gamma}\right)$. Then by \eqref{eq:di/dp},
\begin{align*}
	\frac{\partial P\left(\tilde{p}_x \leq \frac{1}{2}\right)}{\partial p_x}
	&= \begin{multlined}[t][0.7\linewidth]
	\frac{1}{\gamma\Beta\left(\frac{p_x}{\gamma},\frac{1-p_x}{\gamma}\right)} \int_0^{\frac{1}{2}} t^{\frac{p_x}{\gamma}-1} (1-t)^{\frac{1-p_x}{\gamma}-1} \ln\frac{t}{1-t}\od t\\
	+ \frac{\ibeta_{\frac{1}{2}}\left(\frac{p_x}{\gamma}, \frac{1-p_x}{\gamma}\right)}{\gamma}\left[\psi\left(\frac{1-p_x}{\gamma}\right) - \psi\left(\frac{p_x}{\gamma}\right)\right].
	\end{multlined}
\end{align*}
For $0< t <\frac{1}{2}$, the integral term is always negative as $\ln\frac{t}{1-t} < 0$. Note that when the real part of $z$ is positive then the digamma function has the following integral representation $\psi (z)= \int_0^\infty \left(\frac{e^{-t}}{t} -\frac{e^{-zt}}{1-e^{-t}}\right)\od t$, which is increasing in $z$. So when $\frac{1}{2} \leq p_x < 1$, $\psi\left(\frac{1-p_x}{\gamma}\right) \leq \psi\left(\frac{p_x}{\gamma}\right)$ and therefore $\partial P\left(\tilde{p}_x \leq \frac{1}{2}\right) / \partial p_x < 0$.

When $0 < p_x < \frac{1}{2}$, using $\ibeta_{\frac{1}{2}}\left(\frac{p_x}{\gamma}, \frac{1-p_x}{\gamma}\right) = 1 - \ibeta_{\frac{1}{2}}\left(\frac{1-p_x}{\gamma}, \frac{p_x}{\gamma}\right)$, we can get
\begin{align*}
	\frac{\partial P\left(\tilde{p}_x \leq \frac{1}{2}\right)}{\partial p_x} &= -\frac{\partial}{\partial p_x}\ibeta_{\frac{1}{2}}\left(\frac{1-p_x}{\gamma}, \frac{p_x}{\gamma}\right) \\
	&= \begin{multlined}[t][0.7\linewidth]
	\frac{1}{\gamma\Beta\left(\frac{1-p_x}{\gamma},\frac{p_x}{\gamma}\right)} \int_0^{\frac{1}{2}} t^{\frac{1-p_x}{\gamma}-1} (1-t)^{\frac{p_x}{\gamma}-1} \ln\frac{t}{1-t}\od t\\
	+ \frac{\ibeta_{\frac{1}{2}}\left(\frac{1-p_x}{\gamma}, \frac{p_x}{\gamma}\right)}{\gamma}\left[\psi\left(\frac{p_x}{\gamma}\right) - \psi\left(\frac{1-p_x}{\gamma}\right)\right].
	\end{multlined}
\end{align*}
A similar argument would show that $\partial P\left(\tilde{p}_x \leq \frac{1}{2}\right) / \partial p_x < 0$ for $0 < p_x < \frac{1}{2}$. So $P\left(\tilde{p}_x \leq \frac{1}{2}\right)$ is strictly decreasing in $p_x$ and $P\left(\tilde{p}_x > \frac{1}{2}\right)$ is strictly increasing in $p_x$.

We know that $\bigl|p_x -\frac{1}{2}\bigr| > \bigl|p_y -\frac{1}{2}\bigr|$ can be categorized into four cases:
(i) $p_x > p_y > \frac{1}{2}$, (ii) $p_x < p_y < \frac{1}{2}$, (iii) $p_x > \frac{1}{2} > p_y$ and $p_x +p_y >1$, and (iv) $p_y > \frac{1}{2} > p_x$ and $p_x +p_y <1$, where we ignore the weak inequality case for the sake of simplicity.
In the first case, $P\bigl(\bigl(\tilde{p}_x - \frac{1}{2}\bigr)\bigl(p_x -\frac{1}{2}\bigr)>0\bigr) = P\bigl(\tilde{p}_x > \frac{1}{2}\bigr)> P\bigl(\tilde{p}_y >\frac{1}{2}\bigr)=P\bigl(\bigl(\tilde{p}_y - \frac{1}{2}\bigr)\bigl(p_y -\frac{1}{2}\bigr)>0\bigr)$. In the second case, $P\bigl(\bigl(\tilde{p}_x -\frac{1}{2}\bigr)\bigl(p_x -\frac{1}{2}\bigr)>0\bigr) = P\bigl(\tilde{p}_x < \frac{1}{2}\bigr)> P\bigl(\tilde{p}_y < \frac{1}{2}\bigr)=P\bigl(\bigl(\tilde{p}_y-\frac{1}{2}\bigr)\bigl(p_y-\frac{1}{2}\bigr)>0\bigr)$.
In the third case, let $p_z = 1-p_y$, then it is straightforward to check that $P\bigl(\tilde{p}_z > \frac{1}{2}\bigr) = P\bigl(\tilde{p}_y < \frac{1}{2}\bigr)$. Since $p_x > p_z > \frac{1}{2}$, $P\bigl(\bigl(\tilde{p}_x-\frac{1}{2}\bigr)\bigl(p_x -\frac{1}{2}\bigr)>0\bigr) = P\bigl(\tilde{p}_x > \frac{1}{2}\bigr) > P\bigl(\tilde{p}_z >\frac{1}{2}\bigr)=P\bigl(\tilde{p}_y < \frac{1}{2}\bigr) = P\bigl(\bigl(\tilde{p}_y -\frac{1}{2}\bigr)\bigl(p_y -\frac{1}{2}\bigr)>0\bigr)$.
In the fourth case, let $p_z = 1-p_x$, then $P\bigl(\tilde{p}_z > \frac{1}{2}\bigr) = P\bigl(\tilde{p}_x < \frac{1}{2}\bigr)$. Since $p_z > p_y >\frac{1}{2}$, $P\bigl(\bigl(\tilde{p}_x-\frac{1}{2}\bigr)\bigl(p_x-\frac{1}{2}\bigr)>0\bigr) = P\bigl(\tilde{p}_x < \frac{1}{2}\bigr)=P\bigl(\tilde{p}_z > \frac{1}{2}\bigr)>P\bigl(\tilde{p}_y>\frac{1}{2}\bigr)= P\bigl(\bigl(\tilde{p}_y-\frac{1}{2}\bigr)\bigl(p_y -\frac{1}{2}\bigr)>0\bigr)$. This completes the proof.
\end{proof}

\subsection{Lemma \ref{lem:expected-win} and its proof}
\label{pf:expected-win}

We follow the notations in subsection \ref{sec:bidding-strategy}: let $\mathcal{X}_{B+n}^t$ denote the collection of all signal sets $x$ that contain a first black ball in setting $t\in\{pri,soc\}$, and define $V_{B+n}^{i,t} = \sum_{x\in\mathcal{X}_{B+n}^t} P^i(x\mid B)\mathbb{E}\max(\tilde{p}_x^i, 1-\tilde{p}_x^i)$. The following lemma shows that $V_{B+n}^{i,t}$ 
is linear in $\tilde{p}_B^i$ and is also increasing in $\tilde{p}_B^i$.

\begin{lemma}
$V_{B+n}^{i,t} = V_1\tilde{p}_B^i + V_2(1-\tilde{p}_B^i)$ and $\frac{1}{2} < V_2 < V_1 <1$, where $V_1= \sum_{x\in\mathcal{X}_{B+n}^t} P^i(x\mid \urnone, B) \cdot\mathbb{E}\max(\tilde{p}_x^i, 1-\tilde{p}_x^i)$ and $V_2= \sum_{x\in\mathcal{X}_{B+n}^t} P^i(x\mid \urntwo, B)\cdot \mathbb{E}\max(\tilde{p}_x^i, 1-\tilde{p}_x^i)$.
\label{lem:expected-win}
\end{lemma}
\begin{proof}
First,
\begin{align*}
V_{B+n}^{i,t}  &=  \sum\nolimits_{x\in\mathcal{X}_{B+n}^t} [P^i(x, \urnone\mid B) + P^i(x, \urntwo\mid B)]\mathbb{E}\max(\tilde{p}_x^i, 1-\tilde{p}_x^i) \\
&= \sum\nolimits_{x\in\mathcal{X}_{B+n}^t} [\tilde{p}_B^i P^i(x \mid \urnone, B) + (1-\tilde{p}_B^i) P^i(x \mid \urntwo, B)]\mathbb{E}\max(\tilde{p}_x^i, 1-\tilde{p}_x^i) \\
&=\tilde{p}_B^i V_1 + (1-\tilde{p}_B^i)V_2.
\end{align*}

It remains to show that $\frac{1}{2} < V_2 < V_1 <1$. Let $v(p_x) \equiv \mathbb{E}\max(\tilde{p}_x^i, 1-\tilde{p}_x^i)$.
By Proposition \ref{prop:subjective-assessment}, $v(p_x)$ is symmetric about $\frac{1}{2}$ and strictly increasing on $\bigl(\frac{1}{2},1\bigr)$. We shall first prove $V_1 > V_2$ by cases.

\begin{enumerate}[label=(\roman*)]
\item Private information, $n = 1$

In this case, $V_1 = p v(p_\BB) + q v(p_\BW)$ and $V_2 = qv(p_\BB) + pv(p_\BW)$. Since $p>q$, $\frac{p^2}{p^2+q^2} > \frac{1}{2}$. It follows from the monotonicity of $v(\wc)$ that
\begin{equation*}
	V_1 - V_2 = (p-q)[v(p_\BB) - v(p_\BW)] = (p-q)\left[v\left(\frac{p^2}{p^2+q^2}\right) - v\left(\frac{1}{2}\right)\right] >0.
\end{equation*}

\item Private information, $n = 3$

In this case, $V_1 = p^3 v(p_{4B}) + 3p^2q v(p_{3B1W}) + 3pq^2 v(p_{2B2W}) + q^3 v(p_{1B3W})$, $V_2 = q^3 v(p_{4B}) + 3pq^2 v(p_{3B1W}) + 3p^2q v(p_{2B2W}) + p^3 v(p_{1B3W})$. Since $p>q$, $p^3-q^3 > 0$, $3pq(p-q) > 0$, and
\begin{gather*}
	p_{4B} = \frac{p^4}{p^4+q^4} > \frac{p^2}{p^2+q^2} > \frac{1}{2} > \frac{q^2}{p^2+q^2} = p_{1B3W}, \quad
	p_{3B1W}=\frac{p^2}{p^2+q^2} > \frac{1}{2} = p_{2B2W}.
\end{gather*}
Thus, by Proposition \ref{prop:subjective-assessment}, $v(p_{1B3W}) < v(p_{4B})$, $v(p_{2B2W}) < v(p_{3B1W})$, and therefore
\begin{equation*}
	V_1-V_2 = (p^3-q^3)[v(p_{4B})-v(p_{1B3W})] + 3pq(p-q)[v(p_{3B1W})-v(p_{2B2W})] > 0.
\end{equation*}

\item Social information, $n = 1$

In the social information setting, under Assumption \ref{ass:distribution-other-p}, subject $i$'s belief about observing $C_1$ in the true state of urn 1 is\footnote{Since $\tilde{p}_x^{-i}$ is assumed to be continuous, we skip the tie cases where $\tilde{p}_B^{-i} = \frac{1}{2}$ or $\tilde{p}_W^{-i} = \frac{1}{2}$.}
\begin{align}
	P^{i}(C_1\mid\text{urn 1}) &= P(B\mid\text{urn 1})\cdot P^{i}\left(\tilde{p}_B^{-i} > \frac{1}{2}\right) + P(W\mid\text{urn 1})\cdot P^{i}\left(\tilde{p}_W^{-i} > \frac{1}{2}\right) \notag \\
	&= p\left[1-\ibeta_{\frac{1}{2}}\left(\frac{p_B}{\theta^i},\frac{1-p_B}{\theta^i}\right)\right] + q\left[1-\ibeta_{\frac{1}{2}}\left(\frac{p_W}{\theta^i},\frac{1-p_W}{\theta^i}\right)\right] \notag \\
	&= p\ibeta_{\frac{1}{2}}\left(\frac{q}{\theta^i},\frac{p}{\theta^i}\right) + q\ibeta_{\frac{1}{2}}\left(\frac{p}{\theta^i},\frac{q}{\theta^i}\right)
	\label{eq:pi}
\end{align}
as $p_B = p$ and $p_W = q$. Similarly, $P^{i}(C_1\mid\text{urn 2}) = P(B\mid\text{urn 2})\cdot P^{i}\left(\tilde{p}_B^{-i} > \frac{1}{2}\right) + P(W\mid\text{urn 2})\cdot P^{i}\left(\tilde{p}_W^{-i} > \frac{1}{2}\right) = q\ibeta_{\frac{1}{2}}(q/\theta^i,p/\theta^i) + p\ibeta_{\frac{1}{2}}(p/\theta^i,q/\theta^i)$,
$P^{i}(C_2\mid\text{urn 1}) = 1-P^{i}(C_1\mid\text{urn 1})$, and $P^{i}(C_2\mid\text{urn 2}) = 1-P^{i}(C_1\mid\text{urn 2})$.

We omit the index $i$ in the remaining proof for the sake of exposition. Define $\pi \equiv P(\one \mid\urnone) = P(\two\mid\urntwo)$, then in this case, $V_1 = P(\one\mid\urnone) v(p_{B\one}) + P(\two\mid\urnone)v(p_{B\two}) = \pi v(p_{B\one}) + (1-\pi) v(p_{B\two})$ and $V_2 = P(\one\mid\urntwo)v(p_{B\one}) + P(\two\mid\urntwo) v(p_{B\two}) = (1-\pi)v(p_{B\one}) + \pi v(p_{B\two})$.

Because $p>q$ and $\theta > 0$, we have%
\footnote{$\ibeta_{\frac{1}{2}}\left(q/\theta, p/\theta\right) = 1$ and $\ibeta_{\frac{1}{2}}\left(p/\theta, q/\theta\right) = 0$ if $\theta = 0$. But $\theta = 0$ corresponds to the private information case which has been discussed previously.}
\[
    \frac{1}{2}<\ibeta_{\frac{1}{2}}\left(\frac{q}{\theta}, \frac{p}{\theta}\right)< 1, \quad 0 < \ibeta_{\frac{1}{2}}\left(\frac{p}{\theta},\frac{q}{\theta}\right) = 1-\ibeta_{\frac{1}{2}}\left(\frac{q}{\theta},\frac{p}{\theta}\right) < \frac{1}{2}.
\]
Then by \eqref{eq:pi}, $\frac{1}{2} < \pi < p$. This implies that $p_{B\two} = \frac{p(1-\pi)}{p(1-\pi)+q\pi} > \frac{1}{2}$, and
\[
	p_{B\one } = \frac{p\pi}{p\pi + q(1-\pi)} = \frac{1}{1+\frac{q(1-\pi)}{p\pi}}
	> \frac{1}{1+\frac{q\pi}{p(1-\pi)}} = p_{B\two}.
\]
It follows that $v(p_{B\one}) > v(p_{B\two})$ and therefore $V_1 - V_2 = (2\pi - 1) [v(p_{B\one})-v(p_{B\two})] > 0$.

\item Social information, $n = 3$

In this case, $V_1 = \pi^3 v(p_{B3\one}) + 3\pi^2 (1-\pi) v(p_{B2\one 1\two}) + 3\pi(1-\pi)^2 v(p_{B1\one 2\two}) + (1-\pi)^3 v(p_{B3\two})$ and $V_2 = (1-\pi)^3 v(p_{B3\one}) + 3\pi(1-\pi)^2 v(p_{B2\one 1\two}) + 3\pi^2(1-\pi) v(p_{B1\one 2\two}) + \pi^3 v(p_{B3\two})$.

Since $p>q$ and $\pi > \frac{1}{2}$,
\[
    p_{B3\one} = \frac{p\pi^3}{p\pi^3+q(1-\pi)^3} = \frac{1}{1+\frac{q}{p}\bigl(\frac{1-\pi}{\pi}\bigr)^3} > \frac{1}{1+\frac{q}{p}\bigl(\frac{\pi}{1-\pi}\bigr)^3} = \frac{p(1-\pi)^3}{p(1-\pi)^3+q\pi^3} = p_{B3\two}
\]
and
\[
    p_{B3\one} = \frac{1}{1+\frac{q}{p}\bigl(\frac{1-\pi}{\pi}\bigr)^3} > \frac{1}{1+\frac{p}{q}\bigl(\frac{1-\pi}{\pi}\bigr)^3} = \frac{q\pi^3}{p(1-\pi)^3+q\pi^3} = 1-p_{B3\two}.
\]
In addition, note that
\[
    p_{B2\one1\two}= p_{B\one} = \frac{p\pi}{p\pi+q(1-\pi)} > p_{B1\one2\two} = p_{B\two} = \frac{p(1-\pi)}{p(1-\pi)+q\pi} > \frac{1}{2}.
\]
So we have $v(p_{B3\one}) > v(p_{B3\two}) = v(1-p_{B3\two})$ and $v(p_{B2\one1\two}) > v(p_{B1\one2\two})$ by Proposition \ref{prop:subjective-assessment}. Thus
\begin{multline*}
    V_1-V_2 = [\pi^3-(1-\pi)^3][v(p_{B3\one})-v(p_{B3\two})] \\
    + 3\pi(1-\pi)(2\pi-1)[v(p_{B2\one1\two})-v(p_{B1\one2\two})] > 0.
\end{multline*}
\end{enumerate}

Note that in the private learning settings, $V_1$ and $V_2$ are either weighted averages
of $v(p_\BB)$ and $v(p_\BW)$ or weighted averages of $v(p_{4B})$, $v(p_{3B1W})$,
$v(p_{2B2W})$ and $v(p_{1B3W})$. And in the social learning settings, $V_1$ and $V_2$ are
either weighted averages of $v(p_{B\one})$ and $v(p_{B\two})$ or weighted averages
of $v(p_{B3\one})$, $v(p_{B2\one1\two})$, $v(p_{B1\one2\two})$ and $v(p_{B3\two})$.
Because for any signal set $x$, $\frac{1}{2} \leq \max(\tilde{p}_x, 1-\tilde{p}_x) \leq 1$,
\begin{equation}
	\frac{1}{2}\leq v(p_x) \leq 1. \label{eq:1/2<v<1}
\end{equation}
We have $V_1,V_2 \in \bigl[\frac{1}{2}, 1\bigr]$ in any case. However, the first
equality of \eqref{eq:1/2<v<1} holds if and only if $p_x = \frac{1}{2}$ and $\gamma = 0$,
and the second equality of \eqref{eq:1/2<v<1} holds if and only if $p_x = 1$ or 0 and $\gamma = 0$.
It is impossible that all $v(\wc)$'s being weighted are 0, 1, or $\frac{1}{2}$ at the same time. Therefore, $V_1,V_2\neq 1$ and $V_1,V_2\neq\frac{1}{2}$. This completes the proof.
\end{proof}

\subsection{Proof of Proposition \ref{pf:positive-information-value-predict}}
\label{pf:positive-private-info-value}

\begin{proof}
By Lemma \ref{lem:expected-win}, $V_{B+1}^{pri} = \tilde{p}_B^{i} V_1 + (1-\tilde{p}_B^{i})V_2$ with $V_1 = p v(p_\BB) + q v(p_\BW)$ and $V_2 = qv(p_\BB) + p v(p_\BW)$, where $v(p_x) = \mathbb{E}\max(\tilde{p}_x^i, 1-\tilde{p}_x^i)$. Then,
\begin{align*}
    V_{B+1}^{pri} - \max(\tilde{p}_B^i, 1-\tilde{p}_B^i)
    &= \begin{cases} (V_1-V_2-1)\tilde{p}_B^i + V_2  & \text{if } \tilde{p}_B^i > \frac{1}{2},\\
   (1+V_1-V_2)\tilde{p}_B^i + V_2-1  &  \text{if } \tilde{p}_B^i \leq \frac{1}{2}.
    \end{cases}
\end{align*}
Since $\frac{1}{2} < V_2 < V_1 <1$ by Lemma \ref{lem:expected-win}, it must be that
\begin{align*}
    V_{B+1}^{pri} - \max(\tilde{p}_B^i, 1-\tilde{p}_B^i)
    & \begin{cases} >0 & \text{if } \tilde{p}_B^i \in (\ubar{p},\bar{p}),\\
    \leq 0 &  \text{if } \tilde{p}_B^i \geq \bar{p}  \text{ or } \tilde{p}_B^i \leq \ubar{p},
    \end{cases}
\end{align*}
where $\ubar{p} \equiv \frac{1-V_2}{1+V_1-V_2} < \frac{1}{2} < \bar{p} \equiv \frac{V_2}{1-V_1+V_2}$. It remains to show that $\bar{p} > p_B \equiv p$ and $\bar{p}$ is increasing in the Bayesian posterior probability $p$ when $p > \frac{1}{2}$.

Since for any $p_x\in(0,1)$, $v(p_x)=\mathbb{E}\max(\tilde{p}_x, 1-\tilde{p}_x) >\max(\mathbb{E}\tilde{p}_x, 1-\mathbb{E}\tilde{p}_x) =\max(p_x, 1-p_x) \geq p_x$  by Assumption \ref{ass:distribution-p} and Jensen's inequality, then
\begin{align*}
V_1p+V_2(1-p) &= p\left[pv\left(\frac{p^2}{p^2+q^2}\right) + qv\left(\frac{1}{2}\right)\right] + q\left[qv\left(\frac{p^2}{p^2+q^2}\right) + pv\left(\frac{1}{2}\right)\right] \\
&= (p^2+q^2)v\left(\frac{p^2}{p^2+q^2}\right) + 2pq v\left(\frac{1}{2}\right) \\
&> (p^2+q^2)\cdot\frac{p^2}{p^2+q^2} + 2pq\cdot\frac{1}{2} = p,
\end{align*}
so $V_2 > p-V_1p+V_2p$, which implies that $\bar{p} = \frac{V_2}{1-V_1+V_2} > p$.

Note that $\partial \bar{p}/\partial p = \left[V_2\frac{\partial V_1}{\partial p} + (1-V_1)\frac{\partial V_2}{\partial p}\right]/(1-V_1+V_2)^2$.
Since $\frac{p^2}{p^2+q^2}$ is increasing in $p$, and by Assumption \ref{ass:distribution-p}, $v(p_x)$ is an increasing function for $p_x > \frac{1}{2}$, we have $\partial v\left(\frac{p^2}{p^2+q^2}\right)/\partial p > 0$ and
\begin{align*}
\frac{\partial V_1}{\partial p} &= p\cdot\frac{\partial}{\partial p}v\left(\frac{p^2}{p^2+q^2}\right) + v\left(\frac{p^2}{p^2+q^2}\right)-v\left(\frac{1}{2}\right) \\
&> q\cdot\frac{\partial}{\partial p}v\left(\frac{p^2}{p^2+q^2}\right) + \left|v\left(\frac{1}{2}\right) - v\left(\frac{p^2}{p^2+q^2}\right)\right| \\
&> \left|q\cdot\frac{\partial}{\partial p}v\left(\frac{p^2}{p^2+q^2}\right) + v\left(\frac{1}{2}\right) - v\left(\frac{p^2}{p^2+q^2}\right)\right| = \left|\frac{\partial V_2}{\partial p}\right|.
\end{align*}
Since $V_2 > \frac{1}{2} > 1-V_1 > 0$, we have
\[
V_2\frac{\partial V_1}{\partial p}+(1-V_1)\frac{\partial V_2}{\partial p} > (1-V_1)\frac{\partial V_1}{\partial p} - (1-V_1)\left|\frac{\partial V_2}{\partial p}\right| > 0.
\]
This implies that $\partial \bar{p}/\partial p > 0$ and therefore $\bar{p}$ is increasing in $p$.
\end{proof}

\subsection{Proof of Proposition \ref{prop:social-information-value}}
\label{pf:positive-social-info-value}

\begin{proof}
By Lemma \ref{lem:expected-win}, $V_{B+1}^{soc} = \tilde{p}_B^i V_1 + (1-\tilde{p}_B^i) V_2$ with $V_1 = \pi v(p_{B\one})+ (1-\pi) v(p_{B\two})$ and $V_2 = (1-\pi) v(p_{B\one}) + \pi v(p_{B\two})$, where $\pi = P^{i}(\one\mid\urnone) = P^{i}(\two\mid\urntwo)$ and $v(p_x) = \mathbb{E}\max(\tilde{p}_x^i, 1-\tilde{p}_x^i)$. Then applying the same argument as proving Proposition \ref{pf:positive-information-value-predict}, we have
\[
    V_{B+1}^{soc} - \max(\tilde{p}_B^i, 1-\tilde{p}_B^i) \begin{cases}
> 0 & \text{if $\tilde{p}_B^i \in (\ubar{p}', \bar{p}')$}, \\
\leq 0 & \text{if $\tilde{p}_B^i \geq \bar{p}'$ or $\tilde{p}_B^i \leq \ubar{p}'$},
\end{cases}
\]
where $\ubar{p}' \equiv \frac{1-V_2}{1+V_1-V_2} < \frac{1}{2} < \bar{p}' \equiv \frac{V_2}{1+V_2-V_1}$.
Using the property $v(p_x) \geq p_x$ for any $p_x\in (0,1)$, we have
\begin{align*}
	V_1p + V_2(1-p) &= [p\pi+q(1-\pi)]v\left(\frac{p\pi}{p\pi+q(1-\pi)}\right) + [p(1-\pi)+q\pi]v\left(\frac{p(1-\pi)}{p(1-\pi)+q\pi}\right) \\
	&> [p\pi+q(1-\pi)]\cdot\frac{p\pi}{p\pi+q(1-\pi)} + [p(1-\pi)+q\pi]\cdot\frac{p(1-\pi)}{p(1-\pi)+q\pi} \\
	&= p\pi+p(1-\pi) = p.
\end{align*}
It implies that $V_2 > p+V_2p-V_1p \Rightarrow \bar{p}' > p$.
\end{proof}

\subsection{Proof of Proposition \ref{prop:bidding-strategy}}
\label{pf:bidding-strategy}

\begin{proof}
 Since
\begin{align*}
    \frac{\partial U_{n,t}(b,\tilde{p}_B^i)}{\partial b} &\propto (w_0-b+r+\alpha^i )V_{B+n}^t + (w_0-b)(1-V_{B+n}^t) \\
    & \qquad - (w_0+r+\alpha^i)\max(\tilde{p}_B^i, 1-\tilde{p}_B^i) - w_0\left[1-\max(\tilde{p}_B^i, 1-\tilde{p}_B^i)\right] \\
    &= (r+\alpha^i)\left[V_{B+n}^t - \max(\tilde{p}_B^i, 1-\tilde{p}_B^i)\right] -b \\
    &\equiv h^i(b,\tilde{p}_B^i),
\end{align*}
and $\partial^2 U_{n,t}(b,\tilde{p}_B^i)/\partial b^2 = -1/w_0 < 0$, the bid maximizing the subject's expected utility, $b(\tilde{p}_B^i)$, is determined as follows:
\begin{enumerate}[itemsep=0pt, label=(\roman*)]
    \item If $h^i(0,\tilde{p}_B^i) \leq 0$, $b(\tilde{p}_B^i) = 0$;
    \item If $h^i(w_0,\tilde{p}_B^i) \geq 0$, $b(\tilde{p}_B^i) = w_0$;
    \item If $h^i(0,\tilde{p}_B^i) > 0 > h^i(w_0,\tilde{p}_B^i)$, $b(\tilde{p}_B^i)$ is the unique root that satisfies $h^i(b(\tilde{p}_B^i), \tilde{p}_B^i) = 0$, i.e., $b(\tilde{p}_B^i) =(r+\alpha^i)[V_{B+n}^t - \max(\tilde{p}_B^i, 1-\tilde{p}_B^i)]$.
\end{enumerate}
Then the optimal bidding function is $ b(\tilde{p}_B^i) = \min(w_0,\ \max(0,\ (r+\alpha^i)[V_{B+n}^t - \max(\tilde{p}_B^i, 1-\tilde{p}_B^i)])) $. By Lemma \ref{lem:expected-win},
\begin{align*}
    V_{B+n}^t - \max(\tilde{p}_B^i, 1-\tilde{p}_B^i) 
    &= \begin{cases} (V_1-V_2+1)\tilde{p}_B^i-(1-V_2) & \text{if $0\leq\tilde{p}_B^i\leq\frac{1}{2}$}, \\
    (V_1-V_2-1)\tilde{p}_B^i + V_2 & \text{if $\frac{1}{2}< \tilde{p}_B^i\leq 1$}. \end{cases}
    \label{eq:ev-minus-max}
\end{align*}
So in the case of $h^i(0,\tilde{p}_B^i) > 0 > h^i(w_0,\tilde{p}_B^i)$, $b(\tilde{p}_B^i)$ is increasing on $\bigl[0,\frac{1}{2}\bigr]$ and decreasing on $\bigl[\frac{1}{2},1\bigr]$.

We now investigate the boundary situations in more details. Note that $h^i(0,\tilde{p}_B^i) \leq 0$ if and only if $\tilde{p}_B^i \leq \frac{1-V_2}{1+V_1 -V_2}$ when $\tilde{p}_B^i \leq \frac{1}{2}$, and $\tilde{p}_B^i \geq \frac{V_2}{1+V_2 -V_1}$ when $\tilde{p}_B^i > \frac{1}{2}$. Since $\frac{1}{2} < V_2 < V_1<1$, $\frac{1-V_2}{1+V_1 -V_2} < \frac{1}{2}$ and $\frac{V_2}{1+V_2 -V_1} > \frac{1}{2}$. Therefore, $b(\tilde{p}_B^i) = 0$ if and only if $\tilde{p}_B^i \leq \frac{1-V_2}{1+V_1 -V_2}$ or $\tilde{p}_B^i \geq \frac{V_2}{1+V_2 -V_1}$.

In addition, $h^i(w_0,\tilde{p}_B^i) \geq 0$ if and only if $\tilde{p}_B^i \geq \frac{1-V_2 + \Delta}{1+V_1 -V_2}$ when $\tilde{p}_B^i \leq \frac{1}{2}$, and $\tilde{p}_B^i \leq \frac{V_2- \Delta}{1+V_2 -V_1}$ when $\tilde{p}_B^i > \frac{1}{2}$, where $\Delta = w_0/(r+\alpha^i)$. Note that
\[ \frac{1-V_2 + \Delta}{1+V_1 -V_2} > \frac{1}{2} \ \Leftrightarrow\ \frac{V_2- \Delta}{1+V_2 -V_1} < \frac{1}{2} \ \Leftrightarrow\ \alpha^{i} < \frac{2w_0}{V_1 + V_2 -1} -r. \]
Hence, when $\alpha^{i} \le  \frac{2w_0}{V_1 + V_2 -1} -r$, for $\tilde{p}_B^i \in \Bigl[\frac{1- V_2}{1+V_1 - V_2}, \frac{V_2}{1+V_2 - V_1}\Bigr]$, $h^i(w_0,\tilde{p}_B^i) < 0$ and $b(\tilde{p}_B^i) =(r+\alpha^i)\left[V_{B+n}^t - \max(\tilde{p}_B^i, 1-\tilde{p}_B^i)\right]$. When $\alpha^{i} > \frac{2w_0}{V_1 + V_2 -1} -r$,  $b(\tilde{p}_B^i) =w_0$ \\
for $\tilde{p}_B^i \in \Bigl[\frac{1 - V_2 + \Delta}{1+ V_1 -V_2}, \frac{V_2-\Delta}{1+V_2 -V_1}\Bigr]$. This completes the proof.
\end{proof}

\bibliographystyle{ecta}
\bibliography{ref}
\end{document}